\documentclass[11pt]{article}
\usepackage{times}
\usepackage[margin=1in]{geometry}
\usepackage{color}
\usepackage{amsmath,amsthm}
\usepackage{url}
\usepackage{amsfonts}
\usepackage{graphicx}
\usepackage{enumerate}
\usepackage{algorithm}
\usepackage{algorithmic}
\usepackage{hyperref}
\newcommand{\prob}[1]{{\sf Pr}\ensuremath{\left[ #1 \right]}}

\newtheorem{theorem}{Theorem}
\newtheorem*{theorem*}{Theorem}
\newtheorem{lemma}[theorem]{Lemma}
\newtheorem*{lemma*}{Lemma}

\newtheorem{definition}{Definition}

\newtheorem*{observation*}{Observation}
\newtheorem{corollary}[theorem]{Corollary}
\newcommand{\parens}{T}
\newcommand{\dy}{{\sc Dyck}}
\newcommand{\pq}{{\sc PQ}}
\newcommand{\stack}{{\sc Stack}}
\newcommand{\deque}{{\sc Deque}}
\newcommand{\queue}{{\sc Queue}}
\newcommand{\barparens}{\bar{T}}

\newcommand{\hide}[1]{}

\newcommand{\ed}[2]{{\sc Edit}\, #1\,#2\,}

\newcommand{\stred}{{\sc StrEdit}}

\title{Efficiently Computing Edit Distance to Dyck Language}
\author{Barna Saha\\
AT\&T Research Laboratory,\\
 Florham Park, NJ 07932. \\
 Email: \textrm{barna@research.att.com}.\\
 {\bf Full Version} }

\date{}

\begin{document}

\maketitle
\setcounter{page}{0}
\begin{abstract}
Given a string $\sigma$ over alphabet $\Sigma$ and a grammar $G$ defined over the same alphabet, how many minimum number of repairs: insertions, deletions and substitutions are required to map $\sigma$ into a valid member of $G$ ? We investigate this basic question in this paper for \dy$(s)$. \dy$(s)$ is a fundamental context free grammar representing the language of well-balanced parentheses with $s$ different types of parentheses and has played a pivotal role in the development of theory of context free languages. It is also known a nondeterministic version of \dy$(s)$ is the hardest context free grammar. Computing edit distance to \dy$(s)$ has numerous applications ranging from repairing semi-structured documents such as XML to memory checking, automated compiler optimization, natural language processing etc. The problem also significantly generalizes string edit distance which has seen extensive developments over the last two decades and has attracted much attention in theoretical computer science as well as in computational biology community.

It is possible to develop a dynamic programming to exactly compute edit distance to \dy$(s)$ that runs in time cubic in the string length. Such algorithms are not scalable. In this paper we give the {\em first near-linear time algorithm for edit distance computation to \dy$(s)$ that achieves a nontrivial (poly-logarithmic) approximation factor}. In fact, given there exists an algorithm for computing string edit distance on input of size $n$ in $\alpha(n)$ time with $\beta(n)$-approximation factor, we can devise an algorithm for edit distance problem to \dy$(s)$ running in $\tilde{O}(n+\alpha(n))$ \footnote{$\tilde{O}(n)=O(n poly\log{n})$} and achieving an approximation factor of $O(\beta(n)(\log{OPT})^{1.5})$. In $\tilde{O}(n^{1+\epsilon}+\alpha(n))$ time, we get an approximation factor of $O(\frac{1}{\epsilon}\beta(n)\log{OPT})$ (for any $\epsilon > 0$) . Here $OPT$ is the optimal edit distance to \dy$(s)$. Since the best known near-linear time algorithm for string edit distance problem has $\beta(n)=poly\log{n}$, we get the desired bound. Therefore, with the current state of the art, string and \dy$(s)$ edit distance can both be computed within  poly-logarithmic approximation factor in near-linear time. This comes as a surprise since \dy$(s)$ is a significant generalization of string edit distance problem and their exact computations via dynamic programming show a marked difference in time complexity.

Rather less surprisingly, we show that the framework for efficiently approximating edit distance to \dy$(s)$  can be utilized for many other languages. We illustrate this by considering various memory checking languages such as \stack, \queue, \pq~and \deque~which comprise of valid transcripts of stacks, queues,  priority queues and double-ended queues respectively. Therefore, any language that can be recognized by these data structures, can also be repaired efficiently by our algorithm.

%Unlike string edit distance where a symbol at the $i$th position in the reference string must be matched to symbol in position $\in [i+d, i-d]$ in the given string and $d$ being the optimum edit distance, such locality property does not hold for \dy$(s)$. This makes the problem significantly harder.
%
%We therefore focus on developing fast near linear time algorithm by returning approximate answers. Unlike string edit distance where a symbol at the $i$th position in the reference string must be matched to symbol in position $\in [i+d, i-d]$ in the given string and $d$ being the optimum edit distance, such locality property does not hold for \dy$(s)$. This makes the problem significantly harder. However
%paper we propose the first near linear time algorithm (randomized) that on input of length $n$ runs in $\tilde{O}(n+\alpha(n))$ time and provides an approximation factor of $O(\log{d}\beta)$ where

 %We develop a randomized algorithm that runs in linear time and returns a solution with an approximation factor of $O(d)$, where $d$ is again the minimum number of edits required to map $\sigma$ to \dy$(s)$. This implies an $O(\sqrt{n})$-approximation in the worst case where $n$ is the string length. The algorithm is simple and provide a general framework which can be applied to solve edit distance problems to many languages beyond \dy$(s)$. We illustrate this by considering various memory checking languages such as \stack, \queue, \pq and \deque~ which comprise of valid transcripts of stacks, queues,  priority queues and double-ended queues respectively.
%

\end{abstract}

\newpage
\clearpage
\setcounter{page}{1}
\sloppy
\section{Introduction}
In this work, we consider the complexity of {\em language edit distance problem}. That is given $\sigma$ over alphabet $\Sigma$ and a grammar $G$ again over the same alphabet, what is the time complexity of computing the minimum number of edits: insertions, deletions and substitutions, required to convert $\sigma$ into a valid string of $G$ ? We ask this question with respect to \dy$(s)$. \dy$(s)$ is a fundamental context free grammar representing the language of well-balanced parentheses of $s$ different types. For example, string such as ``(()())'' belongs to \dy$(1)$, but ``((())'' does not and string such as ``([])'' belongs to \dy$(2)$ but ``([)]'' does not.

Dyck language appears in many contexts. These languages often describe a property that should be held by commands in most
commonly used programming languages, as well as various subsets of commands/symbols used in latex. Variety of semi-structured data from XML documents to
JSON data interchange files to annotated linguistic corpora contain open and close tags that must be properly nested. They are frequently
massive in size and exhibit complex structures with arbitrary levels of nesting tags (an XML document often encodes an entire database). For example, dblp.xml has current size of $1.2$ GB, is growing rapidly, with $2.3$ million articles that results in a string of parentheses of length more than $23$ million till date.
Furthermore, Dyck language plays an important role in the theory of context-free languages. As stated by the Chomsky-Schotzenberger Theorem, every context free
language can be mapped to a restricted subset of \dy$(s)$ \cite{bnf}. A comprehensive description of context free languages and Dyke languages can be found in \cite{h:78,Kozen:1997}.

The study of {\em language edit distance problem} dates back to early seventies. Such an algorithm for context free grammar was first proposed by Aho and Peterson that runs in $O(|G|^2n^3)$ time where $|\sigma|=n$ is the string length and $|G|$ is the size of the grammar \cite{aho:72}. This was later improved by Myers to run in $O(|G|n^3)$ time \cite{myers:95}. These works were motivated by developing automated parsers for compiler design. For \dy$(s)$ the following dynamic programming improves the running time to $O(n^3)$ independent of $|G|$ which is $\Theta(s)$ in this case. The second condition in computing $Edit[i,j]$ leads to $O(n^3)$ running time as the $jth$ parenthesis $\sigma_j$ (considering it to be a close parenthesis) can be matched with any $\sigma_k$ (considering it to be an open parenthesis), for $k \in [1,j-1]$.
$$Edit[i,j] \leftarrow \min{(cost[i,j]+Edit[i+1,j-1], \min_{i \leq k < j}{(Edit[i,k]+Edit[k+1,j])})}, i \leq j$$
$$cost[i,j]=\mbox{cost of editing $\sigma_i$ to match $\sigma_j$ and is $1$ if $i=j$},$$
$$ Edit[i,j]=\mbox{ edit distance of the substring $\sigma_i\sigma_{i+1}...\sigma_j$ and is $1$ if $i=j$}.$$

Even a decade back, \cite{myv:95} reported these algorithms with cubic running time to be prohibitively slow for parser design. With modern data deluge, the issue of scalability has become far more critical. Motivated by a concrete application of repairing semi-structured documents where imbalanced parenthesis nesting is one of the major reported errors ($14\%$ of XML errors on the web is due to malformedness\cite{gm:2013}) and lack of scalability of cubic time algorithms, the authors in \cite{kssy:13} study the problem of approximating edit distance computation to \dy$(s)$. Given any string $\sigma$ if $\sigma'=\min_{x \in \dy(s)} \ed(\sigma,x)$, they ask the question whether it is possible to design an algorithm that runs in near-linear time and returns $\sigma'' \in \dy(s)$ such that $\ed(\sigma, \sigma'') \leq \alpha \stred(\sigma, \sigma')$ for some $\alpha \geq 1$ where \textsc{StrEdit}~is the normal string edit distance function and $\alpha$ is the approximation factor. Edit distance computation from a string of parentheses to \dy$(s)$ is a significant generalization of string edit distance computation \footnote{For string edit distance computation, between string $\sigma_1$ and $\sigma_2$ over alphabet $C$, create a new alphabet $T \cup \bar{T}$ by uniquely mapping each character $c \in C$ to a new type of open parenthesis, say $t_c$, that now belongs to $T$. Let $\bar{t_c}$ be the matching close parenthesis for $t_c$ and we let $\bar{t_c} \in \bar{T}$. Now create strings $\sigma'_1$ by replacing each character of $\sigma_1$ with its corresponding open parenthesis in $T$, and create string $\sigma'_2$ by replacing each character of $\sigma_2$ with its corresponding close parenthesis in $\bar{T}$. Obtain $\sigma$ by appending $\sigma'_1$ with reverse of $\sigma'_2$. It is easy to check the edit distance between $\sigma$ and \dy$(s)$ is exactly equal to string edit distance between $\sigma_1$ and $\sigma_2$.}. A prototypical dynamic programming for string edit distance computation runs in quadratic time (as opposed to cubic time for \dy$(s)$ edit distance problem). There is a large body of works on designing scalable algorithms for approximate string edit distance computation \cite{bjkk:04,bfc:06,or:07,ak:09,ako:10}. Though basic in its appeal, nothing much is known for approximately computing edit distance to \dy$(s)$.

% Modern Web browsers such as Internet Explorer, Firefox employ rule-based heuristics to rectify mismatching open and closed tags. One simple such rule is to
%substitute a matching close-tag whenever the current close-tag
%does not match the open-tag immediately before it. However, a single extra or missing close-tag is enough to set off a cascade,
%requiring many close-tags to be replaced (or deleted). For example, \verb^<A> <B> <C> <D> </E> </D> </C> </B> </A>^
%would require four substitutions, followed by the deletion of \verb^</A>^.
%Another commonly used rule is to insert a matching close-tag
%whenever the current close-tag does not match, but this can trigger
%a similar cascade: the example above would require four insertions,
%followed by five deletions of the original close-tags.
In \cite{kssy:13}, the authors proposed fast greedy and branch and bound methods with various pruning strategies to approximately answer the edit distance computation to \dy$(s)$. With detailed experimental analysis, they show the effectiveness of these edit distance based approaches in practice over rule based heuristics commonly employed by modern web browsers like Internet Explorer, Firefox etc. However, either their algorithms have worst-case approximation ratio as bad as $\Theta(n)$ or has running time exponential in $d$, where $d$ is the optimal edit distance to \dy$(s)$ (see \cite{kssy:13} for worst case examples).
It is to be noted that for \dy$(1)$, there exists a simple single pass algorithm to compute edit distance: just pair up matching open and close parentheses and report the number of parentheses that could not be matched in this fashion.

 In this paper, we study the question of approximating edit distance to \dy$(s)$ for any $s \geq 2$ and give the {\em first} near-linear time algorithm with nontrivial approximation guarantees. Specifically, given there exists an algorithm for computing string edit distance on input of size $n$ in $\alpha(n)$ time with $\beta(n)$-approximation factor, we can devise an algorithm for edit distance problem to \dy$(s)$ running in $\tilde{O}(n+\alpha(n))$ and achieving an approximation factor of $O(\beta(n)(\log{OPT})^{1.5})$. In $\tilde{O}(n^{1+\epsilon}+\alpha(n))$ time, we get an approximation factor of $O(\frac{1}{\epsilon}\beta(n)\log{OPT})$ (for any $\epsilon > 0$) . Here $OPT$ is the optimal edit distance to \dy$(s)$. Since the best known near-linear time algorithm for string edit distance problem has $\beta(n)=poly\log{n}$, we overall get a poly-logarithmic approximation. Therefore, with the current state of the art, both string and \dy$(s)$ edit distance computation has poly-logarithmic approximation factor in near-linear time. This comes as a surprise since \dy$(s)$ is a significant generalization of string edit distance problem and their exact computations via dynamic programming show a marked difference in time complexity.

 Any parentheses string $\sigma$ can be viewed as $Y_1X_1Y_2X_2....Y_zX_z$ for some $z \geq 0$ where $Y_i$s and $X_i$s respectively consist of only consecutive open and close parentheses. The special case of string edit
 distance problem can be thought of having $z=1$. The approximation factor of our algorithm is in fact $O(\beta(n)\log{z})$, where it is possible to ensure $z \leq OPT$ by a simple preprocessing. Our algorithm cleverly
 combines a random walk to guide selection of subsequences of the parentheses strings in multiple phases having only consecutive sequence of open parentheses followed by close parentheses. These subsequences are then repaired each by employing subroutine for {\sc StrEdit} computation.
 The general framework of the algorithm and its analysis applies to languages far beyond \dy$(s)$. We discuss this connection with respect to several memory checking languages whose study was initiated in the seminal work of Blum, Evans, Gemmell, Kannan and Naor \cite{bwgkn:91}, and followed up by several groups \cite{a:02,ckm:07,dnrv:09,nr:09, cckm:focs10}.
 We consider basic languages such as \stack, \queue, \pq, \deque~etc. They comprise of valid transcripts of stacks, queues, priority queues and double-ended queues respectively. Given a transcript of any such memory-checking language, we consider the problem of finding the minimum number of edits required to make the transcript error-free and show that the algorithm for \dy($s$) can be adapted to return a valid transcript efficiently. Therefore, any language that can be recognized by these data structures, can also be repaired efficiently by our algorithm.

% Any parenthesis string $\sigma$ can be viewed as $Y_1X_1Y_2X_2....Y_zX_z$ where each $Y_i$ is a sequence of open parenthesis and each $X_i$ is a sequence of close parenthesis. String edit distance is a special case where $z=1$. We next extend our first algorithm to provide an approximation guarantee of $O(z\beta)$, where $\beta$ is the best known approximation guarantee for string edit distance. If the string edit distance algorithm runs in $O(n^{a})$ time, then our algorithm running time is $O(n\log{d},n^{a})$.  It can be shown that $\Theta(z)$ is always bounded by $d$, but cases where it is much smaller than $d$, we get a significant improvement over the first algorithm.

%Next, we show how to improve this bound under various scenarios. Finally, we consider several memory checking languages, \stack, \pq, \deque~ whose study was initiated by Blum, Evans, Gemmell, Kannan and Naor \cite{bwgkn:91}, followed up by several groups \cite{a:02,ckm:07,dnrv:09,nr:09} and also considered under the framework of streaming language recognition \cite{cckm:focs10}. \stack, \pq, and \deque~ comprise of valid transcripts of stacks, priority queues and double-ended queues respectively. Given a transcript of any such memory-checking language, we consider the problem of finding the minimum number of edits required to make the transcript error-free and design scalable near-linear time algorithm for returning a valid transcript.

\subsection{Related Work}

%The work of Magniez et al. \cite{mmn:stoc10} followed by \cite{cckm:focs10} consider \dy language recognition problem in the streaming setting. Similar recognition problem has also been studied from ``property testing'' view point \cite{akns:siam01,prr:rand03}. However, recognizing \dy~ language in linear time is easy and can be solved by a single scan using textbook stack-based algorithm. Krebs et al. extended the work of \cite{mmn:stoc10} to consider very special cases of errors, where an open (close) parenthesis can only be changed into another open (close) parenthesis \cite{kls:mfcs11}. Again with only such edits, computing edit distance in linear time is trivial: whenever an open parenthesis at the stack top cannot be matched with the current close parenthesis, change one of them to match. Allowing arbitrary insertion, deletion and substitution makes the problem much harder.
Early works on edit distance to grammar \cite{aho:72,myers:95} was motivated by the problem of correcting and recovering from syntax error during context-free parsing and have received significant attention in the realm of compiler optimization \cite{Fischer:1980, Fischer:1992, myv:95,Kim:2001}. Many of these works focus on designing time-efficient parsers using local recovery \cite{Fischer:1992,myv:95,Kim:2001} rather than global dynamic programming based algorithms \cite{aho:72,myers:95}, but to the best of our knowledge, none of these methods provide approximation guarantee on edit distance in sub-cubic time. Approximating edit distance to \dy$(s)$ has recently been studied in \cite{kssy:13} for repairing XML documents, but again the proposed subcubic algorithms all have worst case approximation factor $\Theta(n)$.

Recognizing a grammar is a much simpler task than repairing. Using a stack, it is straightforward to recognize \dy$(s)$ in a single pass. When there is a space restriction, Magniez, Mathieu and Nayak \cite{mmn:stoc10} considered the streaming complexity of recognizing \dy$(s)$ showing an $\Omega(\sqrt{n})$ lower bound and a matching upper bound within a $\log{n}$ factor. Even with multiple one-directional passes, the lower bound remains at $\Omega(\sqrt{n})$ \cite{cckm:focs10}, surprisingly with two passes in opposite directions the space complexity reduces to $O(\log^{2}{n})$. This exhibits a curious phenomenon of streaming complexity of recognizing \dy$(s)$.  Krebs, Limaye and Srinivasan extended the work of \cite{mmn:stoc10} to consider very restricted cases of errors, where an open (close) parenthesis can only be changed into another open (close) parenthesis \cite{kls:mfcs11}. Again with only such edits, computing edit distance in linear time is trivial: whenever an open parenthesis at the stack top cannot be matched with the current close parenthesis, change one of them to match. Allowing arbitrary insertions, deletions and substitutions is what makes the problem significantly harder. In the property testing framework, in the seminal paper Alon, Krivelevich, Newman and Szegedy \cite{akns:siam01} showed that \dy$(1)$ is testable in time independent of $n$, however \dy$(2)$ requires $\Omega(\log{n})$ queries. This lower bound was further strengthened to $n^{1/11}$ by Parnas, Ron and Rubinfeld \cite{prr:rand03} where they also give a testing algorithm using $n^{2/3}/\epsilon^3$ queries. These algorithms can only distinguish between the case of $0$ error with $\epsilon n$ errors, and therefore, are not applicable to the problem of approximating edit-distance to \dy$(s)$.

 Edit distance to \dy$(s)$ is a significant generalization of string edit distance problem. String edit distance enjoys the special property that a symbol at the $i$th position in one string must be matched with a symbol at a position between $(i-d)$ to $(i+d)$ in the other string. Using this ``local'' property, prototypical  quadratic dynamic programming algorithm can be improved to run in time $O(dn)$ which was later improved to $O(n+d^7)$ \cite{Sahinalp:1996} to $O(n+d^4)$ \cite{Cole:2002} and then to $O(n+d^2)$ \cite{lms:98}. The later result implies a $\sqrt{n}$-approximation for string edit distance problem. However, all of these crucially use the locality property, which does not hold for parenthesis strings: two parentheses far apart can match as well. Also, it is known that parsing arbitrary context free grammar (CFG) is as hard as boolean matrix multiplication \cite{Lee:2002} and a nondeterministic version of \dy~is the hardest CFG \cite{bnf}. Therefore, exactly computing edit distance to \dy$(s)$ in time much less than subcubic would be a significant accomplishment. For string edit distance, the current best approximation ratio of $(\log{n})^{O(\frac{1}{\epsilon})}$ in $n^{1+\epsilon}$ running time for any fixed $\epsilon >0$ is due to Andoni, Krauthgamer and Onak \cite{ako:10}. This result is preceded by a series of works which improved the approximation ratio from $\sqrt{n}$ \cite{lms:98} to $n^{\frac{3}{7}}$ \cite{bjkk:04}, then to $n^{\frac{1}{3}+o(1)}$ \cite{bfc:06}, all of which run in linear time and finally to $2^{\sqrt{\log{n}\log{\log{n}}}}$ that run in time $n2^{\sqrt{\log{n}\log{\log{n}}}}$ \cite{ak:09}.

\subsection{Techniques \& RoadMap}

\begin{definition}
The {\em congruent} of a parenthesis $x$ is defined as its
symmetric opposite parenthesis, denoted $\bar{x}$.
The congruent of a set of parentheses $X$, denoted $\bar{X}$,
is defined as $\{\bar{x} \ | \ x \in X\}$.
\end{definition}
We use $\parens$ to denote the set of {\em open parentheses} and $\barparens$ to denote the set of {\em close parentheses}.
(Each $x \in \parens$ has exactly one congruent $\bar{x} \in \barparens$ that it matches to
and vice versa.) The alphabet $\Sigma=\parens \cup \barparens$.

Consider a string $\sigma=\sigma_1...\sigma_n$, of length $n$, over some
parenthesis alphabet $\Sigma=\parens \cup \barparens$, that is, $\sigma \in \mbox{\dy$(s)$}, s=|T|$.
\begin{definition}
A {\em well-balanced string} over some parenthesis alphabet $\Sigma=\parens \cup \barparens$
obeys the context-free grammar \dy$(s)$, $s=|T|$  with productions
$S \rightarrow SS$, $S \rightarrow \varepsilon$ and
$S \rightarrow a S \bar{a}$ for all $a \in \parens$.
\end{definition}

%\begin{definition}
%The {\em edit distance} $\ed(\sigma,\sigma')$ between two strings $\sigma$ and $\sigma'$
%is the minimum number of insertions, deletions and substitutions
%needed to transform $\sigma$ into $\sigma'$, where
%an insertion of $a$ after position $i$ transforms $\sigma_1 ... \sigma_i \sigma_{i+1} ... \sigma_n$
%to $\sigma_1 ... \sigma_i a \sigma_{i+1} ... \sigma_n$;
%a deletion at position $i$ transforms $\sigma_1 ... \sigma_{i-1} \sigma_i \sigma_{i+1} ... \sigma_n$
%to $\sigma_1 ... \sigma_{i-1} \sigma_{i+1} ... \sigma_n$; and
%a substitution to $a$ at position $i$ transforms $\sigma_1 ... \sigma_{i-1} \sigma_i \sigma_{i+1} ... \sigma_n$
%to $\sigma_1 ... \sigma_{i-1} a \sigma_{i+1} ... \sigma_n$.
%\end{definition}
\begin{definition}
The {\em \dy~ Language Edit Distance Problem},
given string $\sigma$ over alphabet $\Sigma=\parens \cup \barparens$, is to find $\arg\min_{\sigma'} \stred(\sigma,\sigma')$
such that $\sigma' \in \mbox{\dy}(s)$, $s=|\parens|$.
\end{definition}
\subsubsection*{A Simple Algorithm{\small {  (Section \ref{section:algo1})}}}
We start with a very simple algorithm that acts as a stepping stone for the subsequent refinements. The algorithm is as simple as it gets, and we call it
{\bf Random-deletion}.\\

{\em Initialize a stack to empty. Scan the parenthesis string $\sigma$ left to right. If the current symbol is an open parenthesis, insert it into the stack. If the current symbol is a
close parenthesis, check the symbol at top of the stack. If both the symbols can be matched, match them. If the stack is empty, delete the current symbol.
Else {\bf delete one of them independently with equal probability $\frac{1}{2}$}. If the stack is nonempty when the scan has ended, delete all symbols from the stack.}

We show
\begin{theorem}
\label{theorem:single}
 {\bf Random-deletion} obtains a $4d$-approximation for edit distance computation to \dy$(s)$ for any $s \geq 2$ in linear time with constant probability, where
 $d$ is the optimum edit distance.
\end{theorem}

The probability can be boosted by running the algorithm $\Theta(\log{n})$ times and considering the iteration which results in minimum number of edits. In the worst case,
when $d=\sqrt{n}$, the approximation factor can be $4\sqrt{n}$. This also gives a
very simple algorithm for string edit distance problem that achieves an $O(\sqrt{n})$ approximation.

%We give a single pass randomized algorithm that given a string $\sigma$ with constant probability returns a well-formed string by performing $4d^2$ edits, where $d$ is the optimum edit distance from $\sigma$ to \dy$(s)$. The algorithm is very simple. It first deletes all prefixes of close parentheses and all suffixes of open parentheses and matches greedily every well-formed substrings. Next, it performs a stack based comparison and whenever it cannot match the open parenthesis at the stack top with the current close parenthesis on the string, it tosses a coin and deletes one of the two parentheses independently with probability $\frac{1}{2}$. All of these can be implemented in a single scan and the probability can be boosted by running the algorithm $\Theta(\log{n})$ times.

 The analysis of even such a simple algorithm is nontrivial and proceeds as follows. First, we allow the optimal algorithm to consider only deletion and allow it to process the string using a stack; this increases the edit distance by a factor of $2$ (Lemma \ref{lemma:deletion}, Lemma \ref{lemma:opt-stack}). We then define for each comparison where an open and close parenthesis cannot be matched, a corresponding correct and wrong move. If an optimal algorithm  also compares exactly the two symbols and decides to delete one, then we can simply define the parenthesis that is deleted by the optimal algorithm as the correct move and the other as a wrong move. However, after the ``first'' wrong move, the comparisons performed by our procedure vs the optimal algorithm can become very different. Yet, we can label one of the two possible deletions as a correct move in some sense of decreasing distance to an optimal state. Next, we show that if up to time $t$, the algorithm has taken $W_t$ wrong moves then it is possible to get a well-formed string using a total of $2(d+2W_t)$ edit operations. These two properties help us to map the process of deletions to a one dimensional random walk which is known as the {\sc gambler's ruin} problem.

 In gambler's ruin, a gambler enters a casino with $\$d$ and starts playing a game where he wins with probability $1/2$ and loses with probability $1/2$ independently. The gambler plays the game repeatedly betting $\$1$ in each round. He leaves if he runs out of money or gets $\$n$. We can show that the number of steps taken by the gambler to be ruined is an upper bound on the number of edit operations performed by our algorithm. The expected number of such steps is $O(n)$ which is the string length in our case. Interestingly, the underlying probability distribution of the number of steps taken by the random walk is heavy-tailed and using that property, one can still lower bound the probability that gambler is ruined in $O(d^2)$ steps by a constant$\sim\frac{1}{5}$ (Lemma \ref{lemma:2}). This gives an $O(d)$ approximation.
\subsubsection*{A Refined Algorithm{\small {  (Section \ref{section:algo2})}}}

We now refine our algorithm as follows. Given string $\sigma$, we can delete any prefix of close parentheses, delete any suffix of open parentheses and match well-formed substrings without affecting the
optimal solution. After that, $\sigma$ can be written as $Y_1X_1Y_2X_2....Y_zX_z$ where each $Y_i$ is a sequence of open parentheses, each $X_i$ is a sequence of close parentheses and $z \leq d$. In the
optimal solution $X_1$ is matched with some suffix of $Y_1$ possibly after doing edits. Let us denote this suffix by $Z_1$. If we can find the left boundary of $Z_1$, then we can employ
{\sc StrEdit} to compute string edit distance between $Z_1$ and $X_1$ (we have to consider reverse of $X_1$ and convert each $t \in X_1$ to $\bar{t}$--this is what is meant by \stred~between a sequence of open and a sequence of close parentheses), as $Z_1X_1$ consists of only a single
sequence of open parentheses followed by a single sequence of close parentheses. If we can identify $Z_1$ correctly, then in the optimal solution $X_2$ is matched with a suffix of $Y_1^{res}Y_2$ after removing $Y_{1}^{res}=Y_1 \setminus Z_1$.
Let us denote it by $Z_2$. If we can again find the left boundary of $Z_2$, then we can employ {\sc StrEdit} between $Z_2$ and $X_2$ and so on. The question is {\em how do we compute these boundaries ?}

We use {\em Random-deletion} again (repeat it appropriately $\sim \log{n}$ times) to find these boundaries approximately (see Algorithm \ref{alg2}). We consider the suffix of $Y_1$ which {\em Random-deletion} matches against $X_1$ possibly after multiple deletions (call it $Z'_1$) and use $Z'_1$ to approximate $Z_1$.
We show again using the mapping of {\em Random-deletion} to random walk, that the error in estimating the left boundary is bounded (Lemma \ref{lemma:diff}, Lemma \ref{lemma:noofdeletions}, Lemma \ref{lemma:dev}). Specifically, if $\stred(Z_1, X_1)=d_1$, then $\stred(Z'_1,X_1) \leq 2d_1\sqrt{2\log{d_1}}$. Note that this improves from $O(d_1^2)$ by {\em Random-deletion} to $O(d_1\sqrt{\log{d_1}})$. But, the error that we make in estimating $Z_1$ may propagate and affect the estimation of $Z_2$. Hence the gap between optimal $Z_2$ and estimated $Z'_2$ becomes wider. If $\stred(Z_2, X_2)=d_2$, then we get $\stred(Z'_2,X_2) \leq 2(d_1+d_2)\sqrt{2\log{(d_1+d_2)}}$. Proceeding, in this fashion, we get the following theorem.

\begin{theorem}
\label{theorem:single2}
Algorithm \ref{alg2} obtains an $O(z\beta(n)\sqrt{\log{d}})$-approximation factor for edit distance computation  from strings  to \dy$(s)$ for any $s \geq 2$ in $O(n\log{n}+\alpha(n))$ time with probability at least $\left(1-\frac{1}{n}-\frac{1}{d}\right)$, where there exists an algorithm for {\sc StrEdit} running in $\alpha(n)$ time achieves an approximation factor of $\beta(n)$.
\end{theorem}

\subsubsection*{Further Refinement: Main Algorithm {\small {(Section \ref{section:main})}}}
{\em Is it possible to compute subsequences of $\sigma$ such that each subsequence contains a single sequence of open and close parentheses in order to apply \stred, yet propagational error can be avoided ?} This leads to our main algorithm.

\noindent {\bf Example.} Consider $\sigma=Y_1X_1Y_2X_2Y_3X_3...Y_zX_z$, and let the optimal algorithm matches $X_1$ with $Z_{1,1}$, matches $X_2$ with $Z_{1,2}Y_2$, matches $X_3$ with $Z_{1,3}Y_3$, and so on, where $Y_1=Z_{1,z}Z_{1,z-1}...Z_{1,2}Z_{1,1}$. In the {\em refined algorithm}, when {\em Random-deletion} finishes processing $X_1$ and tries to estimate the left boundary of $Z_{1,1}$, it might have already deleted some symbols of $Z_{1,2}$. It is possible that it deletes $\Theta(d_1\sqrt{\log{d_1}})$ symbols from $Z_{1,2}$. Therefore, when computing \stred~between $X_1$ and $Z'_1$, the portion of $Z_{1,2}$ in $Z'_{1}$ may not have any matching
correspondence and results in increased \stred. More severely, the {\em Random-deletion} process gets affected when processing $X_2$. Since, it does not find the symbols in $Z_{1,2}$ which ought to be matched with some subsequence of $X_2$, it penetrates $Z_{1,3}$ and may delete $\Theta((d_1+d_2)\sqrt{\log{(d_1+d_2)}})$ symbols from $Z_{1,3}$, and so on. To remedy this, view $X_2=X_{2,in}X_{2,out}$ where $X_{2,in}$ is the prefix of $X_2$ that is matched with $Y_2$ and $X_{2,out}$ is matched with $Z_{1,2}$. Consider pausing the random deletion process when it finishes $Y_2$ and thus attempt to find $X_{2,in}$. Suppose, {\em Random-deletion} matches $X'_{2,in}$ with $Y_2$, then compute $\stred(Y_2,X'_{2,in})$. While there could be still mistake in computing $X'_{2,in}$, the mistake does not affect $Z_{1,3}$. Else if {\em Random-deletion} process finishes $X_2$ before finishing $Y_2$, then of course it has not been able to affect $Z_{1,3}$. In that case $Z'_{2}$ is a suffix of $Y_2$ and we compute  $\stred(Z'_{2},X_2)$. Similarly, when processing $X_3$, we pause whenever $X_3$ or $Y_3$ is exhausted and create an instance of \stred~accordingly. Suppose, for the sake of this example, $Y_2, Y_3,..,Y_z$ are finished before finishing $X_2,X_3,...,X_z$ respectively and $X_1$ is finished before $Y_1$. Then, after creating the instances of \stred s~as described, we are left with a sequence of open parenthesis corresponding to a prefix of $Y_1$ and a sequence of close parenthesis, which is a combination of suffixes from $X_2,X_3,...,X_z$. We can compute \stred~between them. Since much of $Z_{1,z}Z_{1,z-1}...Z_{1,2}$ exists in this remaining prefix of $Y_1$ and their matching parentheses in the created sequence of close parentheses, the computed \stred~distance will remain within the desired bound.\\

Let us call each $X_iY_i$ a block. As the example illustrates, we create \stred~instances corresponding to what {\em Random-deletion} does locally in each block. After the first phase, from each block we are either left with a sequence of open parenthesis (call it a block of type {\sc O}), or a sequence of close parentheses (call it a block of type {\sc C}), or the block is empty. This creates new sequences of open and close parentheses by combining all the consecutive {\sc O} blocks together (forget empty blocks) and similarly combining all consecutive {\sc C} blocks together (forget empty blocks). We get new blocks, at most $\lfloor \frac{z}{2} \rfloor$ of them after removing any prefix of close and any suffix of open parentheses, in the remaining string. Again, we look at what {\em Random-deletion} does locally in these new blocks and create instances of \stred~accordingly. This process can continue at most $\lceil \log{z} \rceil +1$ phases, since the number of blocks reduces at least by a factor of $2$ going from one phase to the next. This entire process is repeated $O(\log{n})$ time and the final outcome is the minimum of the edit distances computed over these repetitions. The following theorem summarizes the performance of this algorithm.

\begin{theorem}
\label{theorem:main}
There exists an algorithm that obtains an $O(\beta(n)\log{z}\sqrt{\log{d}})$-approximation factor for edit distance computation  from strings  to \dy$(s)$ for any $s \geq 2$ in $O(n\log{n}+\alpha(n))$ time with probability at least $\left(1-\frac{1}{n}-\frac{1}{d}\right)$, where there exists an algorithm for {\sc StrEdit} running in $\alpha(n)$ time that achieves an approximation factor of $\beta(n)$.
\end{theorem}

The $\sqrt{\log{d}}$ factor in the approximation can be avoided if we consider iterating $O(n^{\epsilon}\log{n})$ times. Since, the best known near-linear time algorithm for \stred~anyway has $\alpha(n)=n^{1+\epsilon}$ and $\beta=(\log{n})^{\frac{1}{\epsilon}}$, we obtain the following theorem.

\begin{theorem}
\label{theorem:final}
For any $\epsilon > 0$, there exists an algorithm that obtains an $O(\frac{1}{\epsilon}\log{z}(\log{n})^{\frac{1}{\epsilon}})$-approximation factor for edit distance computation  from strings  to \dy$(s)$ for any $s \geq 2$ in $O(n^{1+\epsilon})$ time with high probability.
\end{theorem}

 The algorithm and its analysis gives a general framework which can be applied to many other languages beyond \dy$(s)$. Employing this algorithm one can repair footprints of several memory checking languages such as \stack, \queue, \pq~ and \deque~ efficiently. We discuss this connection in Section \ref{section:mc}.

\section{Analysis of {\bf Random-deletion}}
\label{section:algo1}
Here we analyse the performance of {\bf Random-deletion} and prove Theorem \ref{theorem:single}.

%In this section we give a single pass algorithm that obtains an $O(d)$ approximation. This implies an $O(\sqrt{n})$ approximation in the
%worst case. We use $o$ to denote open parenthesis and $c$ to denote close parenthesis. The algorithm is given below.
%
%{\begin{algorithm}
%\caption{Computing Edit Distance of $\sigma$ from being well-formed}
%\label{alg1}
%\begin{algorithmic}
%   \scriptsize{\ENSURE $d=0$, $i=1$
%    \WHILE{ $i \leq |\sigma|$}
%    \IF{$\sigma[i]==o$}
%        \STATE Insert $\sigma[i]$ in stack
%        \STATE $i++$
%    \ELSIF{$\sigma[i]$ matches top of stack}
%        \STATE match $\sigma[i]$ with top of stack and remove both of them
%        \STATE $i++$
%    \ELSE
%    \STATE with probability $\frac{1}{2}$ each delete one of $\sigma[i]$ or top of stack
%   \STATE $d=d+1$
%   \IF{$\sigma[i]$ is deleted}
%   \STATE $i++$
%    \ENDIF
%    \ENDIF
%    \ENDWHILE}
%    \end{algorithmic}
%\end{algorithm}}

Recall the algorithm. It uses a stack and scans $\sigma$ left to right. Whenever it
encounters an open parenthesis, the algorithm pushes it into the stack. Whenever it encounters a close parenthesis, it {\em compares} this current symbol with
the one at the stack top. If they can be matched, the algorithm always matches them and proceeds to the next symbol. If the stack is empty, it deletes the close parenthesis and again proceeds to the next symbol. Else, the stack is non-empty but the two parentheses cannot be matched. In that case, the algorithm deletes one of the symbol, either the one at the stack top or the current close parenthesis in the string. It tosses an unbiased coin, and independently with probability $\frac{1}{2}$ it chooses which one of them to delete. If there is no more close parenthesis, but stack is non-empty, then it deletes all the open parentheses from the stack.

We consider only deletion as a viable edit operation and under deletion-only model, assume that the optimal algorithm is stack based, and matches well-formed substrings greedily. The following two lemmas whose proofs are in the appendix state that we lose only a factor $2$ in the approximation by doing so.
 \begin{lemma}
 \label{lemma:deletion}
 For any string $\sigma \in (\parens \cup \barparens)^{*}$, $OPT(\sigma) \leq OPT_{d}(\sigma) \leq 2 OPT(\sigma)$, where $OPT(\sigma)$ is the minimum number of edits: insertions, deletions, substitutions required and $OPT_d(\sigma)$ is the minimum number of deletions required to make $\sigma$ well-formed.
 \end{lemma}

  \begin{lemma}
 \label{lemma:opt-stack}
There exists an optimal algorithm that makes a single scan over the input pushing open parenthesis to stack and when it observes a close parenthesis, it either pops the stack top, if it matches the observed close parenthesis and removes both from further consideration, or edits (that is deletes) either the stack top or the observed close parenthesis whenever there is a mismatch.
\end{lemma}

From now onward we fix a specific optimal stack based algorithm, and refer that as the optimal algorithm.

%\subsection{Analysis}
Let us initiate time $t=0$. At every step in {\em Random-deletion} when we either match two parentheses (current close parenthesis in the string with open parenthesis at the stack top) or delete one of them, we increment the time $t$ by $1$.

We define two sets $A_{t}$ and $A_{t}^{OPT}$ for each time $t$.

\begin{definition}
For every time $t\geq 0$, $A_{t}$ is defined as all the indices of the symbols that are matched or deleted by {\bf Random-deletion} up to and including time $t$.
\end{definition}

\begin{definition}
For every time $t \geq 0$, $$A_{t}^{OPT}=\{i \mid i \in A_{t} \mbox{ or } i \mbox{ is matched by {\bf the optimal algorithm} with some symbol with index in } A_{t}\}.$$
\end{definition}

Clearly at all time $t \geq 0$, $A_{t}^{OPT} \supseteq A_{t}$. We now define a correct and wrong move.

\begin{definition}
A comparison at time $t$ in the algorithm leading to a {\it deletion} is a correct move if $|A_{t}^{OPT} \setminus A_{t}| \leq |A_{t-1}^{OPT} \setminus A_{t-1}|$ and is a wrong move
if $|A_{t}^{OPT} \setminus A_{t}| > |A_{t-1}^{OPT} \setminus A_{t-1}|$.
\end{definition}

\begin{lemma}
At any time $t$, there is always a correct move, and hence {\bf Random-deletion} always takes a correct move with probability at least $\frac{1}{2}$.
\end{lemma}
\begin{proof}
Suppose the algorithm compares an open parenthesis $\sigma[i]$ with a close parenthesis $\sigma[j]$, $i < j$ at time $t$, and they do not match. If possible, suppose that there is no correct move.

Since {\em Random-deletion} is stack-based, $A_{t}$ contains all indices in $[i,j]$. It may also contain intervals of indices $[i_1,j_1], [i_2,j_2],...$ because there can be multiple blocks. It must hold $[1,j-1] \setminus A_{t}$ does not contain any close parenthesis. Now for both the two possible deletions to be wrong,
the optimal algorithm must match $\sigma[i]$ with some $\sigma[j']$, $j' > j$, and also match $\sigma[j]$ with $\sigma[i']$, $i' < i, i' \notin A_{t}$. But, this is not possible due to the property
of well-formedness.

Now consider the case that at time $t$, the stack is empty and the current symbol in the string is $\sigma[j]$. In that case {\em Random-deletion} deletes $\sigma[j]$. Clearly, $A_{t}=[1,j]$.
For this to be a wrong move, the optimal algorithm should match $\sigma[j]$ with $\sigma[i]$, $i < 1$ which is not possible. Hence, in this case the move is correct.

Now consider the case that at time $t$, the input is exhausted and the algorithm considers $\sigma[i]$ from the stack top. In that case {\em Random-deletion} deletes $\sigma[i]$. Clearly
$A_{t}$ contains indices of all close parenthesis. For this to be a wrong move, the optimal algorithm should match $\sigma[i]$ with $\sigma[j]$, $j \notin A_{t}$ which is not possible.
Hence, in this case the move is correct too.
\end{proof}

\begin{lemma}
\label{lemma:diff}
If at time $t$ (up to and including time $t$), the number of indices in $A_{t}^{OPT}$ that the optimal algorithm deletes is $d_{t}$ and the number of correct and wrong
moves are respectively $c_{t}$ and $w_{t}$ then $|A_{t}^{OPT}\setminus A_{t}| \leq d_{t}+w_{t}-c_{t}$.
\end{lemma}
\begin{proof}
Clearly, the lemma holds at time $t=0$, since $d_{0}=c_{0}=w_{0}=0$ and $A_{0}=A_{0}^{OPT}=\emptyset$.

Suppose, the lemma holds up to and including time $t-1$. We now consider time $t$. Let at time $t$, the algorithm compares an open parenthesis $\sigma[i]$ with a close parenthesis $\sigma[j]$, $i < j$. The following cases
need to be considered.

\noindent Case 1. $\sigma[i]$ is matched with $\sigma[j]$. $c_{t}=c_{t-1}, w_{t}=w_{t-1}$.

Subcase 1. The optimal algorithm also matches $\sigma[i]$ with $\sigma[j]$, hence $d_{t}=d_{t-1}$. Now $A_{t}^{OPT}=A_{t-1}^{OPT} \uplus \{i,j\}$ and
$A_{t}=A_{t-1}\uplus \{i,j\}$. So we have $|A_{t}^{OPT}\setminus A_{t}|=|A_{t-1}^{OPT}\setminus A_{t-1}|\leq d_{t-1}+w_{t-1}-c_{t-1} (\mbox{by induction hypothesis})=d_{t}+w_{t}-c_{t}$.

Subcase 2. The optimal algorithm does not match $\sigma[i]$ with $\sigma[j]$. It is not possible that the optimal algorithm matches $\sigma[i]$ with $\sigma[j']$, $j' > j$ and also
matches $\sigma[j]$ with $\sigma[i']$, $i' < j$ simultaneously due to the property of well-formedness.
\begin{itemize}
\item  First consider that $\sigma[i]$ and $\sigma[j]$ are both matched with some symbols in the optimal algorithm. So $d_{t}=d_{t-1}$. Let $\sigma[i]$ be matched with $\sigma[j']$ and $\sigma[j]$ be matched with $\sigma[i']$.
Then either $(a) i < j' < i' < j$, or $(b) i < i' < j < j'$, or $(c) i' < i < j' < j$.
\begin{itemize}
\item For $(a)$ $i,j \in A_{t-1}^{OPT}$ and $A_{t}^{OPT}=A_{t-1}^{OPT}$. Also, $A_{t}=A_{t-1} + \{i,j\}$. Hence $|A_{t}^{OPT}\setminus A_{t}|=|A_{t-1}^{OPT}\setminus A_{t-1}|-2 \leq d_{t-1}+w_{t-1}-c_{t-1}-2=d_{t}+w_{t}-c_{t}-2$.
\item For $(b)$, $j \in A_{t-1}^{OPT}$ and $A_{t}^{OPT}=A_{t-1}^{OPT}\uplus \{i,j'\}$. Also, $A_{t}=A_{t-1} \uplus \{i,j\}$. Hence
$|A_{t}^{OPT}\setminus A_{t}|=|A_{t-1}^{OPT}\setminus A_{t-1}|$ (decreases $1$ due to $j$ and increases $1$ due to $j'$) and thus $|A_{t}^{OPT}\setminus A_{t}|=|A_{t-1}^{OPT}\setminus A_{t-1}|\leq d_{t-1}+w_{t-1}-c_{t-1} (\mbox{by induction hypothesis})=d_{t}+w_{t}-c_{t}$.
\item For $(c)$, $i \in A_{t-1}^{OPT}$ and $A_{t}^{OPT}=A_{t-1}^{OPT}\uplus \{i',j\}$. Also, $A_{t}=A_{t-1} \uplus \{i,j\}$. Hence
$|A_{t}^{OPT}\setminus A_{t}|=|A_{t-1}^{OPT}\setminus A_{t-1}|$ (decreases $1$ due to $i$ and increases $1$ due to $i'$) and thus $|A_{t}^{OPT}\setminus A_{t}|=|A_{t-1}^{OPT}\setminus A_{t-1}|\leq d_{t-1}+w_{t-1}-c_{t-1} (\mbox{by induction hypothesis})=d_{t}+w_{t}-c_{t}$.
\end{itemize}

\item Now consider that one of $\sigma[i]$ or $\sigma[j]$ gets deleted. Assume, w.l.o.g, that $\sigma[i]$ is deleted (exactly similar analysis when only $\sigma[j]$ is deleted). So $d_{t}=d_{t-1}+1$. Then if $\sigma[j]$ is matched with $\sigma[i']$ either $(a') i < i' < j$ or $(b') i' < i < j$.
    \begin{itemize}
    \item For $(a')$, $j \in A_{t-1}^{OPT}$, $A_{t}^{OPT}=A_{t-1}^{OPT} \uplus \{i\}$. Also, $A_{t}=A_{t-1} \uplus \{i,j\}$. Hence
$|A_{t}^{OPT}\setminus A_{t}|=|A_{t-1}^{OPT}\setminus A_{t-1}|-1$ (decreases $1$ due to $j$) and thus $|A_{t}^{OPT}\setminus A_{t}|=|A_{t-1}^{OPT}\setminus A_{t-1}|-1\leq d_{t-1}+w_{t-1}-c_{t-1}-1 (\mbox{by induction hypothesis})< d_{t}+w_{t}-c_{t}$.
    \item For $(b')$, $A_{t}^{OPT}=A_{t-1}^{OPT} \uplus \{i,j,i'\}$. Also, $A_{t}=A_{t-1} \uplus \{i,j\}$. Hence
$|A_{t}^{OPT}\setminus A_{t}|=|A_{t-1}^{OPT}\setminus A_{t-1}|+1$ (increases $1$ due to $i'$) and thus $|A_{t}^{OPT}\setminus A_{t}|=|A_{t-1}^{OPT}\setminus A_{t-1}|+1\leq d_{t-1}+w_{t-1}-c_{t-1}+1 (\mbox{by induction hypothesis})= d_{t}+w_{t}-c_{t}$.
    \end{itemize}
\item Now consider that both $\sigma[i]$ and $\sigma[j]$ are deleted. $d_{t}=d_{t-1}+2$. $A_{t}^{OPT}=A_{t-1}^{OPT} \uplus \{i,j\}$. Also, $A_{t}=A_{t-1} \uplus \{i,j\}$. Hence, $|A_{t}^{OPT}\setminus A_{t}|=|A_{t-1}^{OPT}\setminus A_{t-1}| \leq d_{t-1}+w_{t-1}-c_{t-1}< d_{t}+w_{t}-c_{t}$.
\end{itemize}

\noindent Case 2. $\sigma[i]$ is not matched with $\sigma[j]$ and $\sigma[i]$ is deleted.
\begin{itemize}
\item First consider that in the optimal algorithm, $\sigma[i]$ and $\sigma[j]$ are both matched with some symbols. So $d_{t}=d_{t-1}$. Let $\sigma[i]$ be matched with $\sigma[j']$ and $\sigma[j]$ be matched with $\sigma[i']$.
Then either $(a) i < j' < i' < j$, or $(b) i < i' < j < j'$, or $(c) i' < i < j' < j$.
\begin{itemize}
\item For $(a)$ $i,j \in A_{t-1}^{OPT}$ and $A_{t}^{OPT}=A_{t-1}^{OPT}$. Also, $A_{t}=A_{t-1} \uplus \{i\}$. Hence $|A_{t}^{OPT}\setminus A_{t}|=|A_{t-1}^{OPT}\setminus A_{t-1}|-1 \leq d_{t-1}+w_{t-1}-c_{t-1}-1=d_{t}+w_{t}-c_{t}$, since $c_{t}=c_{t-1}+1$.
\item For $(b)$, $j \in A_{t-1}^{OPT}$ and $A_{t}^{OPT}=A_{t-1}^{OPT}\uplus \{i,j'\}$. Also, $A_{t}=A_{t-1} \uplus \{i\}$. Hence
$|A_{t}^{OPT}\setminus A_{t}|=|A_{t-1}^{OPT}\setminus A_{t-1}|+1$ and thus $|A_{t}^{OPT}\setminus A_{t}|=|A_{t-1}^{OPT}\setminus A_{t-1}|+1\leq d_{t-1}+w_{t-1}-c_{t-1} +1(\mbox{by induction hypothesis})=d_{t}+w_{t}-c_{t}$, since $w_{t}=w_{t-1}+1$.
\item For $(c)$, $i \in A_{t-1}^{OPT}$ and $A_{t}^{OPT}=A_{t-1}^{OPT}$. Also, $A_{t}=A_{t-1} \uplus \{i\}$. Hence
$|A_{t}^{OPT}\setminus A_{t}|=|A_{t-1}^{OPT}\setminus A_{t-1}|-1$ (decreases $1$ due to $i$) and thus $|A_{t}^{OPT}\setminus A_{t}|=|A_{t-1}^{OPT}\setminus A_{t-1}|-1\leq d_{t-1}+w_{t-1}-c_{t-1}-1 (\mbox{by induction hypothesis})=d_{t}+w_{t}-c_{t}$, since $c_{t}=c_{t-1}+1$.
\end{itemize}
\item Now consider that one of $\sigma[i]$ or $\sigma[j]$ gets deleted by the optimal algorithm. Assume, first, that $\sigma[i]$ is deleted. So $d_{t}=d_{t-1}+1$. Then if $\sigma[j]$ is matched with $\sigma[i']$ either $(a') i < i' < j$ or $(b') i' < i < j$.
    \begin{itemize}
    \item For $(a')$, $j \in A_{t-1}^{OPT}$, $A_{t}^{OPT}=A_{t-1}^{OPT} \uplus \{i\}$. Also, $A_{t}=A_{t-1} \uplus \{i\}$. Hence
$|A_{t}^{OPT}\setminus A_{t}|=|A_{t-1}^{OPT}\setminus A_{t-1}|$ and thus $|A_{t}^{OPT}\setminus A_{t}|=|A_{t-1}^{OPT}\setminus A_{t-1}|\leq d_{t-1}+w_{t-1}-c_{t-1} (\mbox{by induction hypothesis})< d_{t}+w_{t}-c_{t}$, since $d_{t}=d_{t-1}+1$, $c_{t}=c_{t-1}+1$ and $w_t=w_{t-1}$.
    \item For $(b')$, $A_{t}^{OPT}=A_{t-1}^{OPT} \uplus \{i\}$. Also, $A_{t}=A_{t-1} \uplus \{i\}$. Hence
Hence $|A_{t}^{OPT}\setminus A_{t}|=|A_{t-1}^{OPT}\setminus A_{t-1}|$ and thus $|A_{t}^{OPT}\setminus A_{t}|=|A_{t-1}^{OPT}\setminus A_{t-1}|\leq d_{t-1}+w_{t-1}-c_{t-1} (\mbox{by induction hypothesis})< d_{t}+w_{t}-c_{t}$, since $d_{t}=d_{t-1}+1$, $c_{t}=c_{t-1}+1$ and $w_t=w_{t-1}$.
    \end{itemize}
    Assume now that the optimal algorithm only deletes $\sigma[j]$. So, $d_{t}=d_{t-1}+1$, because {\em Random-deletion} takes action on $\sigma[i]$ and the optimal algorithm matches it. If $\sigma[i]$ is matched with $\sigma[j']$ either $(a'') i < j' < j$, or $(b'') i < j < j'$.
    \begin{itemize}
    \item For $(a'')$, $i \in A_{t-1}^{OPT}$, $A_{t}^{OPT}=A_{t-1}^{OPT}$ and $A_{t}=A_{t-1} \uplus \{i\}$.Hence $|A_{t}^{OPT}\setminus A_{t}|=|A_{t-1}^{OPT}\setminus A_{t-1}|-1$ and thus $|A_{t}^{OPT}\setminus A_{t}|=|A_{t-1}^{OPT}\setminus A_{t-1}|-1\leq d_{t-1}+w_{t-1}-c_{t-1} -1(\mbox{by induction hypothesis})< d_{t}+w_{t}-c_{t}$, since $c_{t}=c_{t-1}+1$ and $w_{t}=w_{t-1}$.
    \item For $(b'')$, $A_{t}^{OPT}=A_{t-1}^{OPT}\uplus \{i,j'\}$ and $A_{t}=A_{t-1} \uplus \{i\}$.Hence $|A_{t}^{OPT}\setminus A_{t}|=|A_{t-1}^{OPT}\setminus A_{t-1}|+1$ and thus $|A_{t}^{OPT}\setminus A_{t}|=|A_{t-1}^{OPT}\setminus A_{t-1}|+1\leq d_{t-1}+w_{t-1}-c_{t-1}+1(\mbox{by induction hypothesis})< d_{t}+w_{t}-c_{t}$, since $w_{t}=w_{t-1}+1$ and $c_{t}=c_{t-1}$.
    \end{itemize}
\item Now consider that both $\sigma[i]$ and $\sigma[j]$ get deleted by the optimal algorithm. $d_{t}=d_{t-1}+1$, since {\em Random-deletion} takes only action on $\sigma[i]$. $A_{t}^{OPT}=A_{t-1}^{OPT} \uplus \{i\}$. Also, $A_{t}=A_{t-1} \uplus \{i\}$. Hence, $|A_{t}^{OPT}\setminus A_{t}|=|A_{t-1}^{OPT}\setminus A_{t-1}| \leq d_{t-1}+w_{t-1}-c_{t-1}=d_{t}+w_{t}-c_{t}$, since $d_{t}=d_{t-1}+1$ and $c_{t}=c_{t-1}+1$.
\end{itemize}
\noindent Case 3. $\sigma[i]$ is not matched with $\sigma[j]$ and $\sigma[j]$ is deleted.

Same as Case 2.

\noindent Case 4. Now consider the case that at time $t$, the stack is empty and the current symbol in the string is $\sigma[j]$. In that case {\em Random-deletion} deletes $\sigma[j]$ and the move is correct.
So $c_t=c_{t-1}+1$ and $w_{t}=w_{t-1}$. We have $A_{t}=A_{t-1} \uplus \{j\}$. If the optimal algorithm also deletes $\sigma[j]$ then $A_{t}^{OPT}=A_{t-1}^{OPT}\uplus \{j\}$ and $d_{t}=d_{t+1}+1$. Hence
 $|A_{t}^{OPT}\setminus A_{t}|=|A_{t-1}^{OPT}\setminus A_{t-1}| \leq d_{t-1}+w_{t-1}-c_{t-1}(\mbox{by induction hypothesis})= d_{t}+w_{t}-c_{t}$, since $d_{t}=d_{t-1}+1, c_{t}=c_{t-1}+1$. On the other hand,
 if the optimal algorithm matches $\sigma[j]$ with some $\sigma[i']$, $i' < j$, then $j \in A_{t-1}^{OPT}$ and $A_{t}^{OPT}=A_{t-1}^{OPT}$.  Hence
 $|A_{t}^{OPT}\setminus A_{t}|=|A_{t-1}^{OPT}\setminus A_{t-1}|-1 \leq d_{t-1}+w_{t-1}-c_{t-1}-1(\mbox{by induction hypothesis})= d_{t}+w_{t}-c_{t}-2$, since $d_{t}=d_{t-1}, c_{t}=c_{t-1}+1$.

\noindent Case 5. Now consider the case that at time $t$, the input is exhausted and the algorithm considers $\sigma[i]$ from the stack top. In that case {\em Random-deletion} deletes $\sigma[i]$. Then, again by
similar analysis as in the previous case, the claim is established.
\end{proof}
Let $S_{t}$ denote the string $\sigma$ at time $t$ after removing all the symbols that were deleted by {\bf Random-deletion} up to and including time $t$.
\begin{lemma}
\label{lemma:a}
Consider $d$ to be the optimal edit distance. If at time $t$ (up to and including $t$), the number of indices in $A_{t}^{OPT}$ that the optimal algorithm deletes be $d_{t}$ and $|A_{t}^{OPT}\setminus A_{t}| =r_{t}$, at most $r_{t}+(d-d_{t})$ edits is sufficient to convert $S_t$ into a well-balanced string.
\end{lemma}
\begin{proof}
Since $|A_{t}^{OPT}\setminus A_{t}| =r_{t}$, there exists exactly $r_{t}$ indices in $S_{t}$ such that if those indices are deleted, the resultant string is same as what the optimal algorithm obtains after processing the symbols in $A_{t}^{OPT}$. For the symbols in remaining $\{1,2,...,n\} \setminus A_{t}^{OPT}$, the optimal algorithm does at most $d-d_{t}$ edits. Therefore a total of $r_{t}+(d-d_{t})$ edits is sufficient to convert $S_t$ into a well-balanced string.
\end{proof}

\begin{lemma}
\label{lemma:b}
The edit distance between the final string $S_{n}$ and $\sigma$ is at most $d+2w_{n}$.
\end{lemma}

\begin{proof}
Consider any time $t \geq 0$, if at $t$, the number of deletions by the optimal algorithm in $A_{t}^{OPT}$ is $d_{t}$, the number of correct moves and wrong moves are respectively $c_{t}$ and $w_{t}$, then we have
$|A_{t}^{OPT} \setminus A_{t}| \leq d_{t}+w_{t}-c_{t}$. The number of edits that have been performed to get $S_t$ from $S_0$ is $c_{t}+w_{t}$. Denote this by $E(0,t)$. The number of edits that are
required to transform $S_t$ to well-formed is at most $(d-d_t)+d_{t}+w_{t}-c_{t}=d+w_{t}-c_{t}$ (by Lemma \ref{lemma:a}). Denote it by $E'(t,n)$. Hence the minimum total number of edits required (including those already performed) considering state at time $t$ is $E(0,t)+E'(t,n)=d+2w_{t}$. Since this holds for all time $t$, the lemma is established.
\end{proof}

In order to bound the edit distance, we need a bound on $w_{n}$. To do so we map the process of deletions by {\bf Random-deletion} to a random walk.

\subsection{Mapping into Random Walk}
\label{subsec:map}
We consider the following one dimensional random walk. The random walk starts at coordinate $d$, at each step, it moves one step right ($+1$) with probability $\frac{1}{2}$ and moves one step left ($-1$) with probability $\frac{1}{2}$. We count the number of steps required by the random walk to hit the origin.

We now associate a modified random walk with the deletions performed by {\em Random-deletion} as follows. Every time {\em Random-deletion} needs to take a move (performs one deletion), we consider one step of the modified random walk. If {\em Random-deletion} takes a wrong move, we let this random walk make a right (away from origin) step. On the other hand if {\em Random-deletion} takes a correct move,
we let this random walk take a left step (towards origin move). If the random walk takes $W$ right steps, then  {\em Random-deletion} also  makes $W$ wrong moves. If the random walk takes $W$ right steps before hitting the origin, then it takes in total a $d+2W$ steps, and {\em Random-deletion} also deletes $d+2W$ times. Therefore, hitting time of this modified random walk starting from $d$ characterizes the number of edit operations performed by {\em Random-deletion}. In this random walk, left steps (towards origin) are taken with probability $ \geq \frac{1}{2}$ (sometimes with probability $1$). Therefore, hitting time of this modified random walk is always less than the hitting time of an one-dimensional random walk starting at $d$ and taking right and left step independently with equal probability.

We therefore calculate the probability of a one-dimensional random walk taking right or left steps with equal probability  to have a hitting time $D$ starting from $d$. The computed probability serves as a lower bound on the probability that {\em Random-deletion} takes $D$ edit operations to transform $\sigma$ to well-formed.

%If the random walk makes a right step, we let {\bf Random-deletion} to take a wrong move, unless the input is exhausted or stack is empty (in that case, {\bf Random-deletion} takes correct move with probability $1$). Each time the random walk makes a left step, we let the {\bf Random-deletion} to take a correct move. If the random walk takes $W$ right steps, then  {\bf Random-deletion} makes at most $W$ wrong moves. If the random walk takes $W$ right steps before hitting the origin, then it takes in total a $d+2W$ steps, and {\bf Random-deletion} also deletes at most $d+2W$ times. Therefore, by following the steps of random walk and making decisions for {\bf Random-deletion} accordingly, one can obtain a well-formed string using at most the number of steps taken by the random walk to hit the origin. We calculate the probability of the random walk to have a hititng time $D$ starting from $d$. The computed probability serves as a lower bound on the probability that {\bf Random-deletion} takes $D$ edit operations to transform $\sigma$ to well-formed.

The one dimensional random walk is related to {\sc Gambler's Ruin} problem. In gambler's ruin, a gambler enters a casino with $\$d$ and starts playing a game where he wins with probability $1/2$ and loses with probability $1/2$ independently.
The gambler plays the game repeatedly betting $\$1$ in each round. He leaves if he runs out of money or his total fortune reaches $\$N$.
In gambler's ruin problem one is interested in the hitting probability of the absorbing states. For us, we are interested in the probability that
gambler gets ruined. We can set $N=n+d$ because that implies the random walk needs to take $n$ steps at the least to reach fortune and $n$ is a trivial upper bound on the edit distance.

Let $\mathcal{P}_{d}$ denote the probability that gambler is ruined or his fortune reaches $\$N$ on the condition that his current fortune is $\$d$ and also let $\mathcal{E}_{d}$ denote the expected number of steps needed for the gambler to be ruined or get $\$N$ starting from $\$d$.
Then it is easy to see using Markov property that the distribution $\mathcal{P}_{d}$ satisfies the following recursion
$\mathcal{P}_{d}=\frac{1}{2}\mathcal{P}_{d+1}+\frac{1}{2}\mathcal{P}_{d-1}.$ It follows from the above that the expectation $\mathcal{E}_{d}$ satisfies
$\mathcal{E}_{d}=\frac{1}{2}\left(\mathcal{E}_{d+1}+1\right)+\frac{1}{2}\left(\mathcal{E}_{d-1}+1\right)=\frac{1}{2}\left(\mathcal{E}_{d+1}+\mathcal{E}_{d-1}\right)+1.$
Solving the recursion one gets $\mathcal{E}_{d}=d(N-d)$, which is useless in our case. But note that
even though $\mathcal{E}_{1}=N-1$, with probability $\frac{1}{2}$, the gambler is ruined in just $1$ step. This
indicates lack of concentration around the expectation. Indeed, the distribution
$\mathcal{P}_{d}$ is heavy-tailed which can be exploited to bound the hitting time.

We now calculate the probability that the gambler is ruined in $D$ steps precisely.
Let $\mathcal{P}_{d}$ denote the law of a random walk starting in $d \geq 0$, let $\{Y_{i}\}_{0}^{\infty}$ be the i.i.d. steps of the
random walk, let $S_D=d+Y_1+Y_2+...+Y_{D}$ be the position of random walk starting in position $d$ after $D$ steps, and let
$T_{0}=\inf{D : S_D=0}$ denotes the walks first hitting time of the origin. Clearly $T_0=0$ for $\mathcal{P}_{0}$. Then we can show
\begin{lemma}
\label{lemma:1}
For the {\sc Gambler's Ruin} problem $\mathcal{P}_{d}(T_0=D)=\frac{d}{D}\left (
   \begin{array}{c}
   D \\
   \frac{D-d}{2}
   \end{array}
   \right )\frac{1}{2^D}$.
\end{lemma}
\begin{proof}
We first calculate $\mathcal{P}_{d}(S_D=0)$. In order for a random walk to be at position $0$ starting at $+d$, there must be
$r=(D-d)/2$ indices $i_1, i_2,..., i_r$ such that $Y_{i_1}=Y_{i_2}=...=Y_{i_r}=+1$. Rest of the  $\frac{D-d}{2}+d=r+d$ steps must be $-1$.
Hence $\mathcal{P}_{d}(S_D=0)=\left (
   \begin{array}{c}
   D \\
   \frac{D-d}{2}
   \end{array}
   \right )\frac{1}{2^D}$, and the lemma follows from the following hitting time theorem.

   \begin{theorem*}[Hitting Time Theorem\cite{DBLP:journals/tamm/HofstadK08}]
For a random walk starting in $D \geq 1$ with i.i.d. steps $\{Y_{i}\}_{0}^{\infty}$ satisfying that $Y_{i} \geq -1$ almost surely, the distribution of
$T_0$ under $\mathcal{P}_{d}$ is given by
\begin{equation}
\label{appeneqn:2}
\mathcal{P}_{d}\left(T_0=D\right)=\frac{d}{D}\mathcal{P}_{d}\left(S_D=0\right).
\end{equation}
\end{theorem*}
\end{proof}
We now calculate the probability that a gambler starting with $\$d$ hits $0$ within $cd$ steps. Our goal will be to
minimize $c$ as much as possible, yet achieving a significant probability of hitting $0$.
\begin{lemma}
\label{lemma:2}
In {\sc Gambler's Ruin}, the gambler starting with $\$d$ hits $0$ within $2d^2$ steps with probability at least $0.194$.
\end{lemma}
\begin{proof}
From Lemma (\ref{lemma:1})
$\mathcal{P}_{d}(T_0= cd)=  \frac{d}{cd}\left (
   \begin{array}{c}
   cd \\
   \frac{cd-d}{2}
   \end{array}
   \right )\frac{1}{2^{cd}}.$
  % Now for $c=\alpha d$ for some constant $\alpha > 1$, we have
%   $$ \mathcal{P}_{d}(T_0=\alpha d^2)=\frac{1}{\alpha d} \left (
%   \begin{array}{c}
%   \alpha d^2 \\
%   \frac{\alpha d^2-d}{2}
%   \end{array}
%   \right )\frac{1}{2^{\alpha d^2}}.$$
   We now employ the following inequality to bound $\left (
   \begin{array}{c}
   cd \\
   \frac{cd-d}{2}
   \end{array}
   \right )\frac{1}{2^{cd}}$.
   {\small\begin{lemma*}[Lemma 7, Ch. 10 \cite{macwilliams1977theory}](An estimate for a binomial coefficient.) Suppose $\lambda m$ is an integer where $ 0 < \lambda < 1$. Then
   $$\frac{1}{\sqrt{8m\lambda(1-\lambda)}}2^{mH_{2}(\lambda)}\leq \left (
   \begin{array}{c}
   m \\
   \lambda m
   \end{array}
   \right )\leq \frac{1}{\sqrt{2\pi m\lambda(1-\lambda)}}2^{mH_{2}(\lambda)}$$ where $H_{2}$ is the binary entropy function.
   \end{lemma*}}
   We have $\lambda=\frac{1}{2}\left(1-\frac{1}{c}\right)$ and $\left (
   \begin{array}{c}
   cd \\
   \frac{d(c-1)}{2}
   \end{array}
   \right )\geq \frac{2^{cd H_{2}(\frac{1}{2}\left(1-\frac{1}{c}\right))}}{\sqrt{2 cd \left(1-\frac{1}{c}\right) \left(1+\frac{1}{c}\right)}}=\frac{2^{cd H_{2}(\frac{1}{2}\left(1-\frac{1}{c}\right))}}{\sqrt{2 cd \left(1-\frac{1}{c^2}\right) }}.$

   Therefore,
   $\mathcal{P}_{d}(T_0=cd)=\frac{1}{c} \left (
   \begin{array}{c}
   cd \\
   \frac{d(c-1)}{2}
   \end{array}
   \right )\frac{1}{2^{cd}} \geq \frac{1}{c} \frac{2^{cd (H_{2}(\frac{1}{2}\left(1-\frac{1}{c}\right))-1)}}{\sqrt{2 cd \left(1-\frac{1}{c^2}\right) }}.$
   We now use the Taylor series expansion for $H(x)$ around $\frac{1}{2}$, $1-H(\frac{1}{2}-x)= \frac{2}{\ln{2}}(x^2+O(x^{4})).$
Hence,
\begin{align*}
\mathcal{P}_{d}(T_0=cd)&\ge \frac{1}{c \sqrt{2 cd \left(1-\frac{1}{c^2}\right) }} {2^{-\frac{2}{\ln 2} cd \frac{1}{4c^2}}} =  \frac{1}{cd \sqrt{2 \frac{c}{d} \left(1-\frac{1}{c^2}\right) }} {2^{- \frac{1}{2\frac{c}{d}\ln 2}}} >  \frac{e^{- \frac{d}{2c}}}{c \sqrt{2cd }} = \frac{1}{c \sqrt{2cd }e^{ \frac{d}{2c}}},
\end{align*}
Now set $c=\alpha d$, $\alpha > 0$ to get $\mathcal{P}_{d}(T_0=\alpha d^2) \geq \frac{1}{\alpha d^2 \sqrt{2\alpha }e^{ \frac{1}{2\alpha}}}=\frac{A(\alpha)}{d^2}.$
where $A(\alpha) = \frac{1}{\alpha \sqrt{2 \alpha }e^{ \frac{1}{2\alpha}}}$.
Now $A(\alpha)$ is a decreasing function of $\alpha \geq \frac{1}{3}$ (derivative of $\ln{A(\alpha)}$ is negative for $\alpha \geq \frac{1}{3}$).
Therefore, we have, $
\mathcal{P}_{d}( d^2 \le T_0 \le 2 d^2) \ge d^2\frac{A(2)}{d^2} = A(2)=0.194.
$
\end{proof}

\begin{corollary}
\label{corol:1}
In {\sc Gambler's Ruin}, the gambler starting with $\$d$ hits $0$ within $\frac{1}{\epsilon}\frac{d^2}{\log{d}}$ steps for any constant $\epsilon > 0$ with probability at least $\frac{\sqrt{\epsilon\log{d}}}{d^{\epsilon}}$.
\end{corollary}
\begin{proof}
Let $\epsilon'=2\epsilon$. Set $c=\frac{1}{\epsilon'} \frac{d}{\log{d}}$ in the above proof, that is $\alpha=\frac{1}{\epsilon'\log{d}}$ to get $\mathcal{P}_{d}(T_0=\frac{1}{\epsilon'}\frac{d^2}{\log{d}})\geq \frac{(\epsilon' \log{d})^{3/2}}{d^2 \sqrt{2 }e^{ \frac{\epsilon' \log{d}}{2}}}=\frac{(\epsilon' \log{d})^{3/2}}{d^2 \sqrt{2 }}\frac{1}{d^{\epsilon'/2}}.$ Considering $A(\alpha)$ is an increasing function when $\alpha < \frac{1}{3}$, we get
$$
\mathcal{P}_{d}( \frac{1}{\epsilon'}\frac{d^2}{\log{d}} \le T_0 \le \frac{2}{\epsilon'}\frac{d^2}{\log{d}}) \ge \sqrt{\frac{\epsilon'\log{d}}{2}}\frac{1}{d^{\epsilon'/2}}.
$$
Now putting $\epsilon=\frac{\epsilon'}{2}$, we get the result.
\end{proof}

\begin{theorem}
{\bf Random-deletion} obtains a $2d$-approximation for edit distance with deletions only from strings  to \dy$(s)$ for any $s \geq 2$ in linear time with constant probability.
\end{theorem}
\begin{proof}
The theorem follows from the mapping that the edit distance computed by {\em Random-deletion} is at most the number of steps taken by gambler's ruin to hit the origin starting from $\$d$ and then applying Lemma \ref{lemma:2}.
\end{proof}

\begin{theorem*}[\ref{theorem:single}]
{\bf Random-deletion} obtains a $4d$-approximation for edit distance computations (substitution, insertion, deletion) from strings  to \dy$(s)$ for any $s \geq 2$ in linear time with constant probability.
\end{theorem*}
\begin{proof}
Follows from the previous theorem, Lemmma \ref{lemma:deletion} and Lemma \ref{lemma:opt-stack}.
\end{proof}

%In Section~4 in Appendix, we show that the general framework of \rand can be used for many other languages such as basic {\bf memory checking languages} \stack, \queue, \pq, \deque~ etc.

\section{Analysis of the Refined Algorithm}
\label{section:algo2}
We revisit the description of the refined algorithm.
Given a string $\sigma$, we first remove any prefix of close parentheses and any suffix of open parentheses to start with.
Since any optimal algorithm will also remove them, this does not affect the edit distance. We also remove any well-formed substrings since that does not affect the edit distance
as well (Lemma \ref{lemma:opt-stack}).
 Now one can write $\sigma=Y_1X_1Y_2X_2...Y_zX_z$ where each $Y_i, i=1,2,..,z$ consists of only open parentheses and
each $X_i, i=1,2,..,z$ consists of only close parentheses. We call each $Y_iX_i$ a {\em block}. After well-formed substrings removal,
each block requires at least one edit. Hence, we have $z \leq d$ where
$d$ is the optimal edit distance of string $\sigma$.

We can write $\sigma$ based on the processing of the optimal algorithm as follows
$$\sigma=Z_{1,z}Z_{1,z-1}...Z_{1,1}X_1Z_{2,z}Z_{2,z-1}...Z_{2,3}Z_{2,2}X_2......Z_{z-1,z}Z_{z-1,z-1}X_{z-1}Z_{z,z}X_z.$$

Here the optimal algorithm matches $X_1$ with $Z_{1,1}$, that is the close parentheses of $X_1$ are only matched with open parentheses in $Z_{1,1}$ and vice-versa. Some parentheses in $X_1$ and
$Z_{1,1}$ may need to be deleted for matching $X_1$ with $Z_{1,1}$ using minimum number of edits. Similarly, the optimal algorithm matches $X_2$ with $Z_{2,1}Z_{2,2}$, matches
$X_3$ with $Z_{1,3}Z_{2,3}Z_{3,3}$ and so on. Note that it is possible that some of these $Z_{i,j}$ $i=1,2,..,z, j \geq i$ may be empty.

 The algorithm proceeds as follows. It continues {\bf Random-deletion} process as before, but now it keeps track of the substring with which each $X_a$, $a=1,2,..,z$
 is matched (possibly through multiple deletions) during this random deletion process. While processing $X_1$, the random deletion process is restarted $3\log_{b}{n}$ times, $b=\frac{1}{(1-0.194)}=1.24$ and at each time the algorithm
 keeps a count on how many
 deletions are necessary to complete processing of $X_1$. It then selects the particular iteration in which the number of deletions is minimum. We let $Z_{1,min}$ to denote the substring to
 which $X_1$ is matched in that iteration. The algorithm then continues the
 random deletion process. It next stops when processing on $X_2$ finishes. Again, the portion of random deletion process between completion of processing $X_1$ and completion of processing
 $X_2$ is repeated $3\log_{b}{n}$ times and the iteration that results in minimum number of deletions is picked. We define $Z_{2,min}$ accordingly. The algorithm keeps proceeding in a similar manner until the string is exhausted.
 In the process, $Z_{a,min}$ is matched with $X_a$ for $a=1,2,..,z$. But, instead of using the edits that the random deletion process makes to match $Z_{a,min}$ to $X_a$,
 our algorithm invokes the best string edit distance algorithm $\stred(Z_{a,min},X_{a})$ which converts $Z_{a,min}$ to $R_a$ and $X_a$ to $T_a$ such that $R_aT_a$ is well-formed.
 Clearly, at the end we have a well-formed string. The pseudocode of the refined algorithm is given in the appendix (Algorithm \ref{alg2}).

 \subsection{Analysis}
 We first analyze its running time.

 \begin{lemma}
 \label{lemma:runningtime}
 The expected running time of Algorithm \ref{alg2} is $O(n\log{n}+\alpha(n))$ where $\alpha{(n)}$ is the running time of {\sc StrEdit} to approximate string edit distance of input string of length $n$ within factor $\beta(n)$.
 \end{lemma}
 \begin{proof}
 First $\stred(Z_{a,min},X_{a})$, $a=1,2,..,z$ is invoked on disjoint subset of substrings of $\sigma$. Let us denote the length of these substrings by $n_1, n_2,..,n_z$. Since the running tim $\alpha(n)$ is convex, the total time required to run $\stred(Z_{a,min},X_{a})$, $a=1,2,..,z$ is
 $$\alpha(n_1)+\alpha(n_2)+...+\alpha(n_z) \leq \alpha(n_1+n_2+...+n_z)=\alpha(n).$$ We now calculate the total running time of Algorithm \ref{alg2} without considering the running time for \stred. While processing $X_a$, in each of $3\log_{b}{n}$ rounds, the running time is bounded by number of comparisons that our algorithm makes with symbols in $X_a$. Each comparison might result in matching of a symbol in $X_a$--in that case it adds one to the running time, deletion of the symbol in $X_a$, again it adds one to the running time, and deletion of open parenthesis while comparing with a symbol in $X_a$--this might require multiple comparisons with the same symbol of $X_a$. If there are $\eta$ such comparisons with some symbol $x \in X_a$, then there are $\eta-1$ consecutive comparisons with $x$, where {\em random-deletion} chooses to delete the open parenthesis. On the $\eta$th comparison, $x$ is either matched or deleted. Now the probability that {\em random-deletion} deletes $(\eta-1)$ open parenthesis while comparing with $x$ is $\frac{1}{2^{\eta-1}}$. Therefore, the expected number of comparisons involving each symbol $x \in X_a$ before it gets deleted or matched is $\sum_{i\geq1}\frac{i}{2^{i}}\leq 2$. Hence the expected number of total comparisons in each iteration involving $X_a$ is at most $3|X_a|$. Therefore, the total expected number of comparison over all the iterations and $a=1,2,..,z$  is at most $\sum_{a}9|X_a|\log_{b}{n}$. Of course, if the running time becomes more than say 2 times the expected value, we can restart the entire process. By Markov inequality on expectation $2$ rounds is sufficient to ensure that the algorithm will make at most $\sum_{a}18|X_a|\log_{b}{n}$ comparisons in at least one round. Hence the lemma follows.
 \end{proof}

 We now proceed to analyze the approximation factor.
 Let the optimal edit distance to convert $Z_{a}X_{a}$ into well-formed be $d_a$ for $a=1,2,..,z$ where $Z_{a}=Z_{1,a}Z_{2,a}...Z_{a,a}$.

 While computing the set $Z_{a,min}$, it is possible that our algorithm inserts symbols outside of $Z_{a}$ to it or leaves out some symbols of $Z_{a}$.
 {\em In the former case, among the extra symbols that are added, if the optimal algorithm deletes some of these symbols as part of some other $Z_{a'}, a' \neq a$,
 then these deletions are ``harmless''. If we only include these extra symbols to $Z_{a,min}$, then we can as well pretend that those symbols are included in $Z_{a}$ too.
 The edit distance of the remaining substrings are not affected by this modification. Therefore, for analysis purpose, {\bf both for this algorithm and for the main algorithm in the next section}, we always assume w.l.o.g that the optimal algorithm does not delete any of the extra symbols that are added}.

 \begin{lemma}
 \label{lemma:noofdeletions}
 The number of deletions made by random deletion process to finish processing $X_1,X_2,..X_l$, for $l=1,2,..,z$, that is to match $Z_{a,min},X_a$, $a \leq l$, is at most $2(\sum_{a=1}^{l}d_{a})^2$ with probability at least $\left(1-\frac{1}{n^3}\right)^l$.
 \end{lemma}
 \begin{proof}
 Consider $l=1$, that is only $X_1$. The number of deletions made by random deletion process to finish processing $X_1$ is at most the hitting time of an one dimensional random walk starting from position $d_1$. This follows from Section \ref{subsec:map} and noting that any deletion outside of $Z_{1}$ is a wrong deletion. From Lemma \ref{lemma:2}, the hitting time of a random walk starting from $d_1$ is at most $2d_{1}^{2}$ with probability at least $0.194$. Let $b=\frac{1}{(1-0.194)}=1.24$. Since, we repeat the process $3\log_{b}{n}$ times, the probability that the minimum hitting time among these $3\log_{b}{n}$ iterations is more than $2d_{1}^{2}$ is at most  $(1-0.194)^{3\log_{b}{n}}=\frac{1}{n^3}$. Therefore, the number of deletions made by random deletion process to finish processing $X_1$ is at most $2d_{1}^{2}$ with probability $\left(1-\frac{1}{n^3}\right)$.

 Now consider $l=2$, the number of deletions made by random deletion process to finish processing $X_1,X_2$ is at most the hitting time of an one dimensional random walk starting from position $d_1+d_2$. Start the random walk from $(d_1+d_2)$ and first follow the steps as suggested by $Z_{1,min}$ to hit $d_2$ in at most $2d_{1}^{2}$ steps. Now, since we repeat the random deletion process from the time of completing processing of $X_1$ to completing processing of $X_2$, then by similar argument as in $l=1$, the minimum hitting time starting from $d_2$ is at most $2d_{2}^{2}$ with probability $\left(1-\frac{1}{n^3}\right)$. Hence, the minimum hitting time starting from $d_1+d_2$ is at most $2d_{1}^{2}+2d_{2}^{2}< 2(d_1+d_2)^2$ with probability $\left(1-\frac{1}{n^3}\right)^2$.

 Proceeding in a similar manner for $l=3,..,z$, we get the lemma.
 \end{proof}

 Let us denote by $C_l$ and $W_l$ the number of correct and wrong moves taken by the random deletion process when processing of $X_1,X_2,...,X_l$ finishes. Since at each deletion, correct move has been taken with probability at least $\frac{1}{2}$ then by standard Chernoff bound followed by union bound we have the following lemma.

 \begin{lemma}
 \label{lemma:dev}
 When the processing of $X_1,X_2,...,X_l$ finishes $W_l-C_l \leq (2\sum_{a=1}^{l}d_a)\sqrt{2\log{d}}$ with probability at least $\left(1-\frac{1}{n}-\frac{1}{d}\right)$.
 \end{lemma}
 \begin{proof}
Probability that the number of deletions made by random deletion process is at most $2(\sum_{a=1}^{l}d_a)^2$ is $\geq \left(1-\frac{1}{n^3}\right)^l$. Let us denote
 the number of these deletions by $D_l$, for $l=1,2,...,z$. Now by Azuma's inequality
 $$\prob{W_l-C_l > (2\sum_{a=1}^{l}d_a)\sqrt{2\log{d}}| D_1 \leq 2d_1^2 }\leq exp\left( \frac{8(\sum_{a=1}^{l}d_a)^2\log{d}}{4(\sum_{a=1}^{l}d_a)^2}\right)=\frac{1}{d^2}.$$
 Hence $$\prob{W_l-C_l > (2\sum_{a=1}^{l}d_a)\sqrt{2\log{d}}} \leq 1-\left(1-\frac{1}{n^3}\right)^l+\left(1-\frac{1}{n^3}\right)^l\frac{1}{d^2}\leq \frac{l}{n^3}+\frac{1}{d^2}.$$
 Hence $$\prob{\exists l \in [1,z] s.t. W_l-C_l > (2\sum_{a=1}^{l}d_a)\sqrt{2\log{d}}} \leq \frac{z^2}{n^3}+\frac{z}{d^2}\leq \frac{1}{n}+\frac{1}{d}.$$
 \end{proof}

 Now we define $A_{l}^{OPT}$ and $A_{l}$ in a similar manner as in the previous section. We only consider the iterations that correspond to computing $Z_{a,min}$, $a=1,2,..,z$ to define the
 final random deletion process.

 \begin{definition}
$A_{l}$ is defined as all the indices of the symbols that are matched or deleted by {\bf Random-deletion} process up to and including time when processing of $X_l$ finishes.
\end{definition}

\begin{definition}
For every time $l \in [1,z]$, $$A_{l}^{OPT}=\{i \mid i \in A_{l} \mbox{ or } i \mbox{ is matched by {\bf the optimal algorithm} with some symbol with index in } A_{l}\}.$$
\end{definition}

We have the following corollary

\begin{corollary}
\label{cor:diff}
For all $l \in [1,z]$, $|A_{l}^{OPT}\setminus A_{l}| \leq \sum_{a=1}^{l}d_a+(2\sum_{a=1}^{l}d_a)\sqrt{2\log{d}}$ with probability at least $\left(1-\frac{1}{n}-\frac{1}{d}\right)$.
\end{corollary}
\begin{proof}
Proof follows from Lemma \ref{lemma:diff} and Lemma \ref{lemma:dev}.
\end{proof}

\begin{lemma}
\label{lemma:stred}
For all $a \in [1,z]$, $\stred{(Z_{a,min},X_a)} \leq d_a + |A_{a-1}^{OPT}\setminus A_{a-1}|+|A_{a}^{OPT}\setminus A_{a}|$.
\end{lemma}
\begin{proof}
Let $D(X_a)$ denote the symbols from $X_a$ for which the matching open parentheses have already been deleted before processing on $X_a$ started.
Let $D'(X_a)$ denote the symbols from $X_a$ for which the matching open parentheses are not included in $Z_{a,min}$.
Let $E(Z_{a,min})$ denote open parentheses in $Z_{a,min}$ such that their matching close parentheses are in $X_a'$, $a' < a$, that is they are already deleted.
Let $E'(Z_{a,min})$ denote open parentheses in $Z_{a,min}$ such that their matching close parentheses are in $X_a'$, $a' > a$, that is they are extra symbols from higher blocks.
$$\stred{(Z_{a,min},X_a)}=\stred(Z_a,X_a)+|D(X_a)|+|D'(X_a)|+|E(Z_{a,min})|+|E'(Z_{a,min})|.$$
Now all the indices of $D(X_a)\cup E(Z_{a,min})$ are in $A_{a-1}^{OPT}$, but none of them are in $A_{a-1}$. Hence
$$|D(X_a)|+|E(Z_{a,min})|\leq |A_{a-1}^{OPT}\setminus A_{a-1}|.$$
Also, the indices corresponding to matching close parenthesis of $E'(Z_{a,min})$ and $D'(X_a)$ are in $A_{a}^{OPT}$ but not in $A_{a}$. Hence
$$|D'(X_a)|+|E'(Z_{a,min})|\leq |A_{a}^{OPT}\setminus A_{a}|.$$
Therefore, the lemma follows.
\end{proof}

In fact, we can have a stronger version of the above lemma, though it does not help in obtaining a better bound for Theorem \ref{theorem:single2}.

\begin{lemma}
\label{lemma:stred2}
For all $a \in [1,z]$, $\stred{(Z_{a,min},X_a)} \leq d_a + |A_{a-1}^{OPT}\setminus A_{a-1}|+|\{A_{a}^{OPT}\setminus A_{a}\}\setminus \{A_{a-1}^{OPT}\setminus A_{a-1}\}|$.
\end{lemma}
\begin{proof}
Follows from Lemma \ref{lemma:stred} and noting that $$|D'(X_a)|+|E'(Z_{a,min})|\leq |\{A_{a}^{OPT}\setminus A_{a}\}\setminus \{A_{a-1}^{OPT}\setminus A_{a-1}\}|.$$
\end{proof}

\begin{theorem}[\ref{theorem:single2}]
Algorithm \ref{alg2} obtains an $O(z\beta(n)\sqrt{\log{d}})$-approximation factor for edit distance computation  from strings  to \dy$(s)$ for any $s \geq 2$ in $O(n\log{n}+\alpha(n))$ time with probability at least $\left(1-\frac{1}{n}-\frac{1}{d}\right)$, where there exists an algorithm for {\sc StrEdit} running in $\alpha(n)$ time achieves an approximation factor of $\beta(n)$.
\end{theorem}
\begin{proof}
The edit distance computed by the refined algorithm is at most $\sum_{a=1}^{z} \stred(Z_{a,min},X_a)$. Hence by Lemma \ref{lemma:stred} the computed edit distance (assuming we use optimal algorithm for
\stred~) is at most
$$\sum_{a=1}^{z} d_a + |A_{a-1}^{OPT}\setminus A_{a-1}|+|A_{a}^{OPT}\setminus A_{a}|=d+2\sum_{a=1}^{z} |A_{a}^{OPT}\setminus A_{a}|$$
$$\leq 7zd\sqrt{2\log{d}}, \mbox{ by Corollary \ref{cor:diff}.}$$

Since, we use a $\beta(n)$-approximate algorithm for \stred~ and that runs in $\alpha(n)$ time, we get an $O(\beta(n) z\sqrt{\log{d}})$-approximation with running time  $O(n\log{n}+\alpha(n))$ (Lemma \ref{lemma:runningtime}).
Hence the theorem follows.
\end{proof}

\noindent{\bf Note.}
 Instead of repeating different portions of random deletion process and then stitching the random deletions corresponding to $Z_{a,min}, a=1,2,..,z$, we can simply repeat the entire {\em Random-deletion} $\Theta(\log{n})$ time. For each repetition, run the entire algorithm and finally return the edit distance corresponding to the iteration that returns the minimum value. We get the same approximation factor and asymptotic running time by doing this. However, we need to calculate \stred~in each iterations, which is not the case in the described algorithm. Therefore, in practice, we save some running time.

\section{Further Refinement: Main Algorithm \& Analysis}
\label{section:main}
\sloppy
We again view our input $\sigma$ as $Y_1X_1Y_2X_2...Y_zX_z$, where each $Y_a$ is a sequence of open parentheses and each $X_a$ is a sequence of
close parentheses. Each $Y_a,X_a$, $a=1,2,..,z$, is a block. We know by preprocessing, $z \leq d$, where $d$ is the optimal edit distance to \dy$(s)$.

As before, we first run the process of {\bf Random-deletion}. For each run of {\em Random-deletion}, the algorithm proceeds in phases with at most $\lceil \log_{2}{z} \rceil +1$ phases.
We repeat this entire procedure $3\log_{b}{n}$ times, $b=1.24$ (as before) and return the minimum
edit distance computed over these runs and the corresponding well-formed string.
We now describe the algorithm corresponding to a single run of {\em random-deletion} (also shown pictorially in Figure \ref{fig:algo3}).

Let us use $X_{a}^{1}=X_{a}, Y_{a}^{1}=Y_{a}$ to denote the blocks in the first phase.
Consider the part of {\em Random-deletion} from the start of processing $X_a^1$ to finish either $X_a^1$ or $Y_a^1$ whichever happens first.
 Since this part of the random deletion (from the start of $X_a^1$ to the completion of either $X_a^1$ or $Y_a^1$) remains confined within block $Y_a^1X_a^1$, we call this part {\em $local^1$ to block $a$}. Let $A_{a}^{local^1}$ denote the indices of all the symbols that are matched or deleted during the $local^1$ steps in block $a$. Let $A_{a}^{OPT, local^1}$ be the union of $A_{a}^{local^1}$ and the indices of symbols that are matched with some symbol with indices in $A_{a}^{local^1}$ in the optimal solution. We call $A_{a}^{OPT,local^1} \setminus A_{a}^{local^1}$ the $local^1$ error, denoted $local\mbox{-}error^1(Y_a^1,X_a^1)$.

Create substrings $L_{a}^1$ corresponding to $local^1$ moves in block $a$, $a=1,..,z$. Compute {\sc StrEdit} between $L_{a}^1\cap Y_{a}^1$ to $L_a^1 \cap X_a^1$. Remove all these substrings from
further consideration. The phase $1$ ends here.

 We can now view the remaining string as $Y_1^{2}X_1^{2}Y_2^{2}X_2^{2}...Y_{z^{2}}^{2}X_{z^{2}}^{2}$, after deleting any prefix of open parentheses and any suffix of close parentheses.
 Consider any $Y_a^2,X_a^2$. Let they span the original blocks $Y_{a_1}X_{a_1}Y_{a_1+1}X_{a_1+1}...Y_{a_2}X_{a_2}$.
 Consider the part of {\em Random-deletion} from the start of processing $X_{a_1}$ to the completion of either $Y_{a_1}$ or $X_{a_2}$ whichever happens first.
 Since this part of the random deletion remains confined within block $Y_a^2X_a^2$, we call this part {\em $local^2$ to block $a$}. Let $A_{a}^{local^2}$ denote the indices of all the symbols that are matched or deleted during the $local^2$ steps in block $a$. Let $A_{a}^{OPT, local^2}$ be the union of $A_{a}^{local^2}$ and the indices of symbols that are matched with some symbol with indices in $A_{a}^{local^2}$ in the optimal solution. We call $A_{a}^{OPT,local^2} \setminus A_{a}^{local^2}$ the $local^2$ error, denoted $local\mbox{-}error^2(Y_a^2,X_a^2)$.

 Create substrings $L_{a}^2$ corresponding to $local^2$ moves in block $a$, $a=1,..,z^2$. Compute {\sc StrEdit} between $L_{a}^2\cap Y_{a}^2$ to $L_a^2 \cap X_a^2$. Remove all these substrings from
 further consideration.

 We continue in this fashion until the remaining string becomes empty. We can define $local^i$ moves, $A_{a}^{local^i}$, $A_{a}^{OPT,local^i}$ and  $local\mbox{-}error^i(Y_a^i,X_a^i)$ accordingly.

 \begin{definition}
 For $i$th phase blocks $Y_{a}^{i}X_{a}^{i}$, if they span original blocks $Y_{a_1}X_{a_1}Y_{a_1+1}X_{a_1+1}...Y_{a_2}X_{a_2}$, then
 part of random deletion from the start of processing $X_{a_1}$ to finish either $Y_{a_1}$ or $X_{a_2}$ whichever happens first, remains confined in block $Y_{a}^{i}X_{a}^{i}$ and
 is defined as $local^i$ move.
 \end{definition}

 \begin{definition}
 For any $i \in \mathbb{N}$, $A_{a}^{local^i}$ denote the indices of all the symbols that are matched or deleted during the $local^i$ steps in block $a$.
 \end{definition}

 \begin{definition}
For any $i \in \mathbb{N}$, $A_{a}^{OPT, local^i}$ denote the union of $A_{a}^{local^i}$ and the indices of symbols matched with some symbol in $A_{a}^{local^i}$.
% within block $Y_{a}^{i}X_{a}^i$.
 \end{definition}

 \begin{definition}
For any $i \in \mathbb{N}$, $A_{a}^{OPT,local^i} \setminus A_{a}^{local^i}$ is defined as the $local^i$ error, $local\mbox{-}error^i$.
 \end{definition}

We now summarize the algorithm below.

\noindent{\bf Algorithm:}

Given the input $\sigma=Y_1X_1Y_2X_2...Y_zX_z$, the algorithm is as follows
\begin{itemize}
\item $MinEdit=\infty$,
\item For $iteration=1, iteration \leq 3\log_{b}{n}, iteration++$
\begin{itemize}
\item Run {\bf Random-deletion} process.
\item Set $i=1, z^1=z$, $edit=0$, and  for $a=1,2,...,z^1$, $X_{a}^{1}=X_{a}, Y_{a}^{1}=Y_a$.
\item While $\sigma$ is not empty
\begin{itemize}
\item Consider the part of random-deletion from the start of processing $X_{a}^{i}$ to finish either $X_{a}^{i}$ or $Y_{a}^{i}$ whichever happens first.
\item Create substrings $L_{a}^i$, $a=1,2,..,z^i$ which correspond to $local^{i}$ moves. Compute ${\sc StrEdit}(L_{a}^i\cap Y_{a}^i, L_a^i \cap X_a^i)$ to match $L_{a}^i\cap Y_{a}^i$ to $L_a^i \cap X_a^i$ and add the required number of edits to $edit$.
\item Remove $L_a^i$, $a=1,2,..,z^i$ from $\sigma$, write the remaining string as $Y_1^{i+1}X_1^{i+1}Y_2^{i+1}X_2^{i+1}...Y_{z^{i+1}}^{i+1}X_{z^{i+1}}^{i+1}$, possibly by deleting any prefix of close parentheses and any suffix of open parentheses. The number of such deletions is also added to $edit$. Set $i=i+1$
\end{itemize}
\item End While
\item If ($edit < MinEdit$) set $MinEdit=edit$
\item End For
\end{itemize}
\item Return $MinEdit$ as the computed edit distance.
\end{itemize}

Of course, the algorithm can compute the well-formed string by editing the parentheses that have been modified in the process through {\sc StrEdit} operations.

%and by repeating the subprocess enough number of times ensure that
%\begin{enumerate}
%\item Lemma \ref{lemma:noofdeletions} holds.
% The number of deletions made by random deletion process to finish processing $X_1,X_2,..X_l$, for $l=1,2,..,z$, that is to match $Z_{a,min},X_a$, $a \leq l$, is at most $2(\sum_{a=1}^{l}d_{a})^2$ with probability at least $\left(1-\frac{1}{n^3}\right)^l$.
%\item For any $a' \leq b'$, the number of deletions made by random deletion process between the start of processing $X_{a'}$ and finishing either $X_{b'}$ or $Y_{a'}$ whichever happens first is at most $\left(\sum_{i=a'}^{b'} d_{i}\right)^{2}$.
%\end{enumerate}
%
%To ensure property (2), we may need to run {\sc random-deletion} subprocess $c\log_{b}{n}$ times from the point of finishing the computation of some $Y_{a}$ or $X_{a}$ to the next completion of some $Y_{a'}$ or $X_{a'}$. We take $b=\frac{1}{1-0.194}=0.124$, and consider the iteration that gives the minimum number of deletions for this subprocess. The entire random deletion process over the whole string is then
%obtained by gluing these at most $2z$ subprocesses.

\begin{lemma}
 \label{lemma-algo3:noofdeletions}
 There exists at least one iteration among $3\log_{b}{n}$, such that for all $a' \leq b'$, (P1) the number of deletions made by random deletion process between the start of processing $X_{a'}$ and finishing either $X_{b'}$ or $Y_{a'}$ whichever happens first is at most $2d(a',b')^2$, where
 $d(a',b')$ is the number of deletions the optimal algorithm does starting from the beginning of $X_{a'}$ to complete either $X_{b'}$ or $Y_{a'}$ whichever happens first with probability at least $\left(1-\frac{1}{n}\right)$.
 \end{lemma}
 \begin{proof}
 Consider the random source $\mathcal{S}$ that supplies the random coins for deciding which directional deletions to execute, and consider the outcomes of the random source for
 each of the $3\log_{b}{n}$ iterations of {\em Random-deletion}.

 Now consider any $a',b', a \leq b'$. Suppose that the optimal algorithm finishes $Y_{a'}$ first (the argument when the optimal algorithm finishes $X_{b'}$ first is identical), and
 let the processing on $Y_{a'}$ completes while comparing with $X_{b''}$, $a' \leq b'' \leq b'$. Consider this portion of the substring, and let the number of deletions
 performed by the optimal algorithm in this substring be $\gamma$. Then the number of deletions made by the random deletion process to finish processing $Y_{a'}$ is at most the hitting time of an one dimensional random walk starting from position $\gamma$. From Lemma \ref{lemma:2}, the hitting time of a random walk starting from $\gamma$ is at most $2\gamma^2$ with probability at least $0.194$. Let $b=\frac{1}{(1-0.194)}=1.24$. If, we repeat the process $3\log_{b}{n}$ times, then the probability that there exists at least one iteration with hitting time not more than $2\gamma^2$ is at least  $1-(1-0.194)^{3\log_{b}{n}}=1-\frac{1}{n^3}$. Now use the outcomes of the random source that have been applied for repeating the entire random walk corresponding to the portion of the random walk starting from $X_{a'}$ and finishing either $X_{b'}$ or $Y_{a'}$ whichever happens first. This has the same effect as repeating the process only for the portion of the random walk.

 There are at most $z+(z-1)+(z-2)+...+1< z^2$ possible values of $a',b'$. Hence the probability that there exists at least one iteration in which (P1) holds for all $a',b'$ is by union bound at least $1-\frac{z^2}{n^3}=\left(1-\frac{1}{n}\right)$.

 \end{proof}

\begin{figure}[t!]
  \centering
  \includegraphics[width=\textwidth]{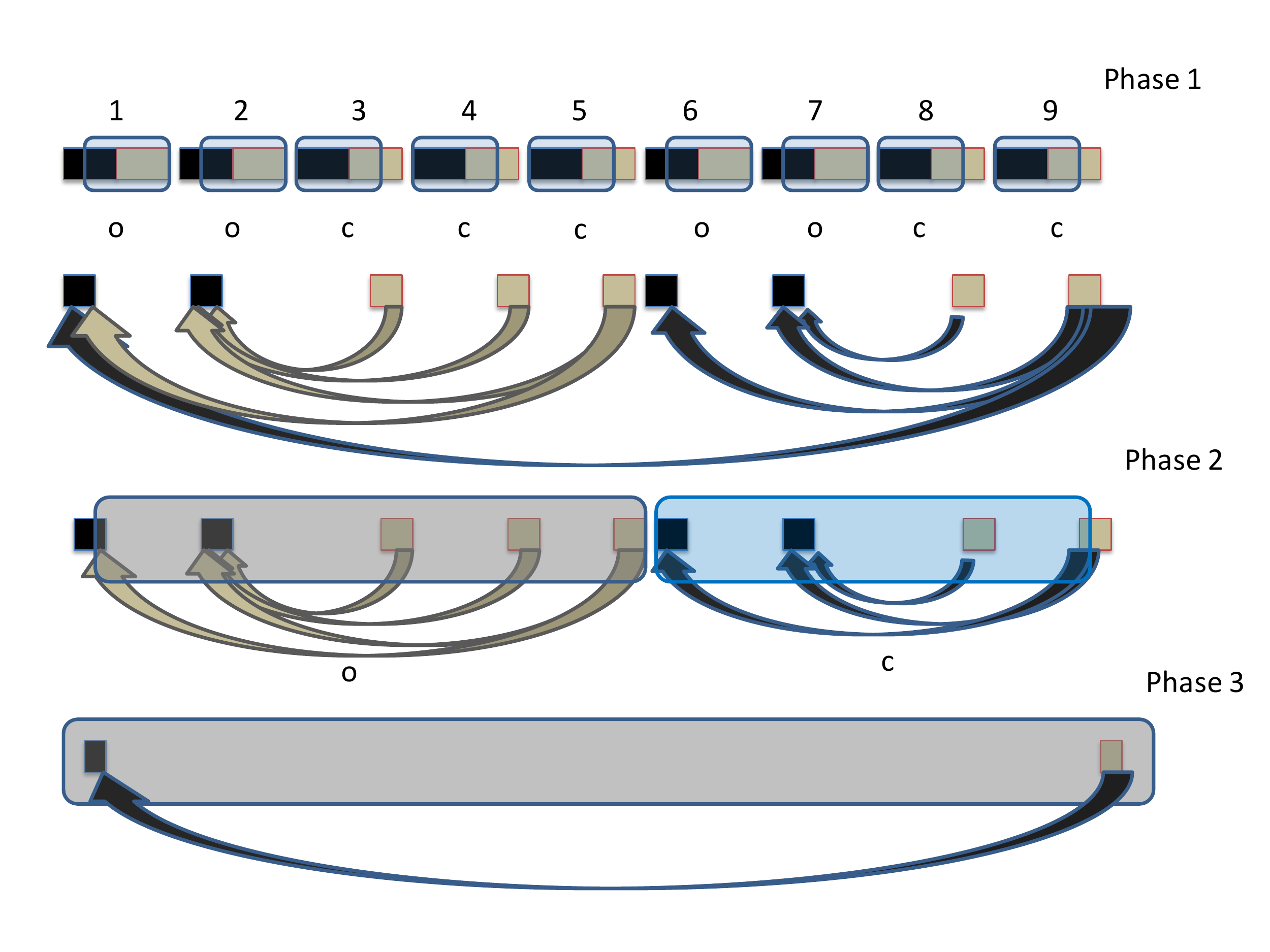}\\
  \caption{Phase by phase execution of the main algorithm}\label{fig:algo3}
\end{figure}

%\begin{definition}
%Define old error for $level^i$ block $a$ as $Old[local^i-a]$ as all symbols from $L_{a}^i$ such that their matching symbols are already deleted before processing on $local^i$ a block starts.
%\end{definition}
%\begin{definition}
%Define new error for $level^i$ block $a$ as $New[local^i-a]$ as all symbols from $L_{a}^i$ such that their matching symbols exist when processing on $local^i$ a block starts but are not part of it when processing on $local^i$ $a$ block ends.
%\end{definition}
%
%\begin{observation}
%$\stred{(L_a^i\cap Y_{a}^i, L_a^i\cap X_{a}^i)} \leq d_a^i+Old[local^i-a]+New[local^i-a]$, where the optimal algorithm matches some subsequence of $X_{a}^i$ with some subsequence of $Y_a^i$ using $d_a^i$ edits.
%\end{observation}
\subsection{Analysis of Approximation Factor}
\begin{lemma}
\label{lemma:phaseL-edit}
Considering phase $l$ which has $z^l$ blocks, the total edit cost paid in phase $l$ of the returned solution is
$\sum_{a=1}^{z^l} \stred{(L_a^l\cap Y_{a}^l, L_a^l\cap X_{a}^l)}  \leq \beta(n)(d + l*d\sqrt{20\log{d}})$ with probability at least $1-\frac{1}{n}-\frac{1}{d}$.
\end{lemma}
\begin{proof}
Consider the iteration in which Lemma \ref{lemma-algo3:noofdeletions} holds, that is we have property (P1). We again fix an optimal stack based algorithm and refer to it as
the optimal algorithm.\\
\noindent{{\bf Phase 1.}}

Consider $Y_{a}^{1},X_{a}^{1}$. We know $Y_{a}^1=Y_a$ and $X_a^1=X_a$.

\begin{itemize}
\item If no symbols of either of $Y_a$ or $X_a$ is matched by the optimal algorithm outside of block $a$ (that is they are matched to each other),
then let $d_a^1$ denote the optimal edit distance to match $Y_a$ with $X_a$.

By Lemma \ref{lemma-algo3:noofdeletions}, the random deletion process takes at most $2(d_a^1)^2$ steps within
$local^1_a$ moves. Now from Lemma \ref{lemma:diff}, local error, $A_{a}^{OPT,local^1} \setminus A_{a}^{local^1} \leq d_{a}^1+W_a^{local^1}-C_{a}^{local^1}$ where $W_a^{local^1}$ is the number of wrong moves taken during $local^1$ steps in block $a$ and similarly $C_{a}^{local^1}$ is the number of correct steps taken during the $local^1$ steps in block $a$. Since the total number of local steps is at most $2(d_{a}^1)^{2}$ and wrong steps are taken with probability at most $\frac{1}{2}$, hence by standard application of Chernoff bound or by Azuma's inequality as in Lemma \ref{lemma:dev},
local error, $local\mbox{-}error^1(Y_a^1,X_a^1) \leq d_{a}^1+d_{a}^1\sqrt{12\log{d_a^1}}\leq  d_{a}^1\sqrt{20\log{d_a^1}}$ with probability at least $1-\frac{1}{d^3}$.

Hence
$\stred{(L_a^1\cap Y_{a}^1, L_a^1\cap X_{a}^1)} \leq d_a+ local\mbox{-}error^1(Y_a^1,X_a^1) \leq d_{a}^1+ d_{a}^1\sqrt{20\log{d_a^1}}$.

\item If some suffix of $X_a^1$ is matched outside of block $a$, then let $X_a^{1,p}$ be the prefix of $X_a^1$ which is matched inside. Consider $Y_a^1X_a^{1,p}$.
Let $d_a^1$ denote the optimal edit distance to match $Y_a$ with $X_a^{p}$. Now again, the total number of local steps is at most $2(d_{a}^1)^{2}$ and wrong steps are taken with probability at most $\frac{1}{2}$, hence by standard application of Chernoff bound or by Azuma's inequality as in Lemma \ref{lemma:dev},
local error, $local\mbox{-}error^1(Y_a^1,X_a^1) \leq d_{a}^1+d_{a}^1\sqrt{12\log{d_a^1}}\leq d_{a}^1\sqrt{20\log{d_a^1}} $ with probability at least $1-\frac{1}{d^3}$.

Hence again
$\stred{(L_a^1\cap Y_{a}^1, L_a^1\cap X_{a}^1)} \leq d_a+ local\mbox{-}error^1(Y_a^1,X_a^1) \leq d_{a}^1+d_{a}^1\sqrt{20\log{d_a^1}}$.

\item If some prefix of $Y_a^1$ is matched outside of block $a$, then let $Y_a^{1,s}$ be the suffix of $Y_a^1$ which is matched inside. Consider $Y_a^{1,a}X_a^{1}$ and let
$d_a^1$ denote the optimal edit distance to match $Y_a^{1,s}$ with $X_a$. Again similar arguments lead to
$\stred{(L_a^1\cap Y_{a}^1, L_a^1\cap X_{a}^1)} \leq d_a^1+ local\mbox{-}error^1(Y_a^1,X_a^1) \leq d_{a}^1+d_{a}^1\sqrt{20\log{d_a^1}}$.
\end{itemize}

Hence due to phase 1, the total edit cost paid is at most $\sum_{a=1}^{z} \stred{(L_a^1\cap Y_{a}^1, L_a^1\cap X_{a}^1)}  \leq \sum_{a=1}^{z} d_{a}^1+d_{a}^1\sqrt{20\log{d_a^1}} \leq d + d\sqrt{20\log{d}}.$\\

\noindent{\bf Phase 2.}

Consider $Y_{a}^2,X_{a}^2$. Let $Y_a^2,X_a^2$ spans original blocks $Y_{g}X_{g}Y_{g+1}X_{g+1}....,Y_{h}X_{h}$ (we drop the superscript $1$ for the original blocks).

\begin{itemize}
\item If both $X_a^2$ and $Y_a^2$ are matched by the optimal algorithm inside $g$ to $h$ blocks, then let $d_{g}^{h}$ denote the optimal edit distance to match $Y_{g}X_{g}Y_{g+1}X_{g+1}....,Y_{h}X_{h}$.
By Lemma \ref{lemma-algo3:noofdeletions}, the random deletion process takes at most $2(d_{g}^{h})^2$ steps within blocks $g$ to $h$. Hence
$local\mbox{-}error^2(Y_a^2,X_a^2)=A_{a}^{OPT,local^2} \setminus A_{a}^{local^2}\leq d_{g}^{h}+ d_{g}^{h}\sqrt{12\log{d_{g}^{h}}}\leq d_{g}^{h}\sqrt{20\log{d_{g}^{h}}}$.
Suppose, {\em Random-deletion} selects $R \subset Y_{a}^{2}$ and $T \subset X_{a}^{2}$ to match. Note that either $R=Y_{a}^{2}$ or $T=X_{a}^{2}$.
Let $D(R,T)$ denote all the symbols in $R,T$ such that their matching parentheses belong to blocks either outside of index $[g,h]$ or they exist at phase 2 but are not included.
Let $E(R,T)$ denote all the symbols in $R,T$ such that their matching parentheses belong to blocks $[g,h]$ but have already been deleted in phase 1.

Then $$\stred(L_a^2\cap Y_a^2, L_a^2 \cap Y_a^2) \leq d_{g}^{h}+|D(R,T)|+|E(R,T)|$$
Now, $$|D(R,T)| \leq  local\mbox{-}error^2(Y_a^2,X_a^2)=A_{a}^{OPT,local^2} \setminus A_{a}^{local^2}\leq d_{g}^{h}\sqrt{20\log{d_{g}^{h}}}$$
For each $x \in E(R,T)$ either its matching parentheses $y$ belongs to the same block in phase $1$, say $a_i\in [g,h]$, or in different blocks, say $a_i$ and $a_{i'} \in [g,h]$.
In the first case, $x$ can be charged to some $y$ (with which it matches) which contributes $1$ to $local\mbox{-}error^1(Y_{a_i}^1,X_{a_i}^1)$. In the second case,
$x$ can again be charged to some $y$ which contributes $1$ to $local\mbox{-}error^1(Y_{a_{i'}}^1,X_{a_{i'}}^1)$. Hence,
$$|E(R,T)| \leq \sum_{i=g}^{h}local\mbox{-}error^1(Y_i^1,X_i^1).$$
Thus, we get
$$\stred(L_a^2\cap Y_a^2, L_a^2 \cap Y_a^2) \leq d_{g}^{h} + local\mbox{-}error^2(Y_a^2,X_a^2)+\sum_{i=g}^{h}local\mbox{-}error^1(Y_i^1,X_i^1)$$
$$\leq d_{g}^{h}+ d_{g}^{h}\sqrt{20\log{d_{g}^{h}}}+ d_{g}^{h}\sqrt{20\log{d_{g}^{h}}}$$
$$\leq d_{g}^{h}+ 2 d_{g}^{h}\sqrt{20 \log{d_{g}^{h}}}$$

\item If some symbol of $X_{a}^2$ is matched outside, let us denote the set of those indices by $\Upsilon$. Consider the subsequence of $X_{a}^{2}$ excluding $\Upsilon$. Let us denote it by $X_{a}^{2,\bar{\Upsilon}}$.
Consider $Y_a^2 X_{a}^{2,\bar{\Upsilon}}$, and proceed as in the previous case.

\item If some symbol of $Y_{a}^2$ is matched outside, let us denote the set of those indices by $\Upsilon$. Consider the subsequence of $Y_{a}^{2}$ excluding $\Upsilon$. Let us denote it by $Y_a^{2,\bar{\Upsilon}}$.
Consider $Y_a^{2,\bar{\Upsilon}} X_a^{2}$, and proceed as in the previous case.
\end{itemize}

Hence due to phase 2, the total edit cost paid is at most $\sum_{a=1}^{z^2} \stred{(L_a^1\cap Y_{a}^1, L_a^1\cap X_{a}^1)}  \leq d + 2 d\sqrt{20 \log{d}}$. \\

\noindent{\bf Phase l.}

Proceeding exactly as in Phase 2, let $Y_a^l$ and $X_a^l$ contains blocks from index $[g^{l-1},h^{l-1}]$ of level $l-1$, all of which together contain blocks $[g^{l-2},h^{l-2}]$ from
level $l-2$ and so on, finally $[g^1,h^1]=[g,h]$ of blocks from level 1. Let $d_{g}^{h}$ denote the optimal edit distance to match $Y_a^l$ and $X_a^l$ excluding the
symbols that are matched outside of blocks $[g,h]$ by the optimal algorithm.

Suppose, {\em Random-deletion} selects $R \subset Y_{a}^{l}$ and $T \subset X_{a}^{l}$ to match. Note that either $R=Y_{a}^{L}$ or $T=X_{a}^{L}$.
Let $D(R,T)$ denote all the symbols in $R,T$ such that their matching parentheses belong to blocks either outside of index $[g,h]$ or they exist at phase $l$ but are not included.
Let $E(R,T)$ denote all the symbols in $R,T$ such that their matching parentheses belong to blocks $[g,h]$ but have already been deleted in phases $1,2,..,l-1$.

Then $$\stred(L_a^l\cap Y_a^l, L_a^l \cap Y_a^l) \leq d_{g}^{h}+|D(R,T)|+|E(R,T)|$$
Now, $$|D(R,T)| \leq  local\mbox{-}error^l(Y_a^l,X_a^l)=A_{a}^{OPT,local^l} \setminus A_{a}^{local^l}\leq d_{g}^{h}+ d_{g}^{h}\sqrt{12\log{d_{g}^{h}}}\leq d_{g}^{h}\sqrt{20\log{d_{g}^{h}}}.$$

For each $x \in E(R,T)$, consider the largest phase $\eta \in [1,2,..,l-1]$ such that its matching parenthesis $y$ existed before start of the $\eta$th phase but does not exist after the end of the
 $\eta$th phase. It is possible, either $y$ belongs to the same block in phase $\eta$, say $r\in [g^{\eta},h^{\eta}]$, or in different blocks, say $r$ and $s \in [g^{\eta},h^{\eta}]$.
In the first case, $x$ can be charged to $y$ which contributes $1$ to $local\mbox{-}error^{\eta}(Y_{r}^{\eta},X_{r}^{\eta})$. In the second case,
$x$ can again be charged to $y$ which contributes $1$ to $local\mbox{-}error^{\eta}(Y_{s}^{\eta},X_{s}^{\eta})$.

Hence,
$$|E(R,T)| \leq \sum_{j=l-1}^{1}\sum_{i=g^{j}}^{h^{j}}local\mbox{-}error^j(Y_i^j,X_i^j).$$
$$\stred(L_a^l\cap Y_a^l, L_a^l \cap Y_a^l) \leq d_{g}^{h} + local\mbox{-}error^l(Y_a^l,X_a^l)+\sum_{j=l-1}^{1}\sum_{i=g^{j}}^{h^{j}}local\mbox{-}error^j(Y_i^j,X_i^j)$$
$$\leq d_{g}^{h}+ d_{g}^{h}\sqrt{20\log{d_{g}^{h}}}+(l-1)* d_{g}^{h}\sqrt{20\log{d_{g}^{h}}} \leq d_{g}^{h}+ ld_{g}^{h}\sqrt{20\log{d_{g}^{h}}}.$$

Hence due to phase $l$, the total edit cost paid is at most $$\sum_{a=1}^{z^l} \stred{(L_a^l\cap Y_{a}^l, L_a^l\cap X_{a}^l)}  \leq d + l*d\sqrt{20\log{d}}.$$

Now, since we are using a $\beta(n)$-approximation algorithm for {\sc \stred}, we get the total edit cost paid during phase $l$ is at most $\beta(n)(d + l*d\sqrt{20\log{d}})$.

For the above bound to be correct, all the $local\mbox{-}error$ estimates have to be correct. The number of blocks reduces by $\frac{1}{2}$ from one phase to the next. Hence, the total number of
 local error estimates is $\Theta(z)$. We have considered the iteration such that property (P1) holds for all sequence of blocks (see Lemma \ref{lemma-algo3:noofdeletions}).
Given (P1) holds, since there are a total of $\Theta(z)$ blocks over all the phases, with probability at least $1-\frac{\Theta(z)}{d^3} > 1-\frac{1}{d}$, all the $local\mbox{-}error$ bounds used in the analysis are
correct. Since, (P1) holds with probability at least $(1-\frac{1}{n})$, with probability at least $\left(1-\frac{1}{n}\right)\left(1-\frac{1}{d}\right)>1-\frac{1}{n}-\frac{1}{d}$, we get a total edit cost paid during phase $l$ is at most $\beta(n)(d + l*d\sqrt{20\log{d}})$.
\end{proof}
\begin{lemma}
\label{lemma:totaleditcost}
The total edit cost paid is at most $O((\log{z})^2 \beta(n)\sqrt{\log{d}})$ with probability at least $1-\frac{1}{n}-\frac{1}{d}$.
\end{lemma}
\begin{proof}
By Lemma \ref{lemma:phaseL-edit}, summing up to and including phase l, the total edit cost paid is at most $O(\beta(n)l^2 d \sqrt{\log{d}})$ with probability $1-\frac{1}{n}-\frac{1}{d}$ when the number of phases is $l$.
Now each phase reduces the number of blocks at least by a factor of $2$. Hence the total number of phases $\leq \lceil \log{z}\rceil +1$. Therefore, total edit cost paid is at most $O(\beta(n)(\log{z})^2 d \sqrt{\log{d}})$.
\end{proof}

\begin{theorem}
There exists an algorithm that obtains an $O(\beta(n)\sqrt{\log{d}}(\log{z})^2)$-approximation factor for edit distance computation  from strings  to \dy$(s)$ for any $s \geq 2$ in $O(n\log{n}+\alpha(n))$ time with probability at least $\left(1-\frac{1}{n}-\frac{1}{d}\right)$, where there exists an algorithm for {\sc StrEdit} running in $\alpha(n)$ time that achieves an approximation factor of $\beta(n)$.
\end{theorem}
\begin{proof}
For a particular iteration, each \stred~is run on a disjoint subsequence. Hence, the running time of the algorithm is $O(n\log_{b}{n}+\alpha(n))$. Therefore, from Lemma \ref{lemma:totaleditcost}, we get the theorem.
\end{proof}

\subsection{Improving the Bound to $O( \beta(n)\log{z}\sqrt{\log{d}})$}
The above argument can be easily strengthened to improve the approximation factor to $O(\beta(n)\log{z} \sqrt{\log{d}})$.

\begin{lemma}
\label{lemma:imptotaleditcost}
The total edit cost paid is at most $O(\beta(n)\log{z}\sqrt{\log{d}})$ with probability at least $1-\frac{1}{n}-\frac{1}{d}$.
\end{lemma}
\begin{proof}
Consider any level $l\geq 1$. let $Y_{g^l}^l$ and $X_{g^l}^{l}$ contains blocks from index $[g^{l-1},h^{l-1}]$ of level $l-1$, all of which together contain blocks $[g^{l-2},h^{l-2}]$ from
level $l-2$ and so on, finally $[g^1,h^1]=[g,h]$ of blocks from level 1. Let $d_{g}^{h}$ denote the optimal edit distance to match $Y_{g^l}^l$ and $X_{g^l}^{l}$ excluding the
symbols that are matched outside of blocks $[g,h]$ by the optimal algorithm.

We want to bound
$$\sum_{j=l}^{1}\sum_{a=g^{j}}^{h^{j}}\stred(L_a^j\cap Y_a^j, L_a^j \cap X_a^j)$$

Let $R_{a}^j=L_a^j\cap Y_a^j$ and $T_{a}^j=L_a^j\cap X_a^j$.
Note that either $R_{a}^j=Y_a^j$ or $T_{a}^j=X_a^j$.

Let $A(Y_a^j)$ indicate the minimum index of original phase-$1$ block that $Y_a^j$ contains and $B(X_a^j)$ indicate the maximum index of original phase-$1$ block that $X_a^j$ contains.
Let $D(R_{a}^j,T_{a}^j)$ denote all the symbols in $R_{a}^j,T_{a}^j$ such that their matching parentheses belong to blocks either outside of index $[A(Y_a^j),B(X_a^j)]$ or they exist at phase $j$ but are not included.
Let $E(R_{a}^j,T_{a}^j)$ denote all the symbols in $R_{a}^j,T_{a}^j$ such that their matching parentheses belong to blocks $[A(Y_a^j),B(X_a^j)]$ but have already been deleted in phases $1,2,..,j-1$.

Then $$\stred(L_a^j\cap Y_a^j, L_a^j \cap X_a^j)=\stred(R_{a}^j,T_{a}^j) \leq d_{A(Y_a^j)}^{B(X_a^j)}+|D(R_{a}^j,T_{a}^j)|+|E(R_{a}^j,T_{a}^j)|$$

where,
$d_{A(Y_a^j)}^{B(X_a^j)}$ denotes the optimal edit distance to match $Y_a^j$ and $X_a^j$ excluding the
symbols that are matched outside of blocks $[A(Y_a^j), B(X_a^j)]$.

Now, by definition

$$|D(R_{a}^j,T_{a}^j)| \leq  local\mbox{-}error^j(Y_a^j,X_a^j)=A_{a}^{OPT,local^j} \setminus A_{a}^{local^j}$$
$$\leq d_{A(Y_a^j)}^{B(X_a^j)}+ d_{A(Y_a^j)}^{B(X_a^j)}\sqrt{12\log{d_{A(Y_a^j)}^{B(X_a^j)}}}\leq d_{A(Y_a^j)}^{B(X_a^j)}\sqrt{20\log{d_{A(Y_a^j)}^{B(X_a^j)}}}.$$

For each $x \in E(R_{a}^j,T_{a}^j)$, consider the largest phase $\eta \in [1,2,..,j-1]$ such that its matching parenthesis $y$ existed before start of the $\eta$th phase but does not exist after the end of the
 $\eta$th phase. It is possible, either $y$ belongs to the same block in phase $\eta$, say the $r$th block, or in different blocks, say the $r$th and the $s$th block.
In the first case, $x$ can be charged to $y$ which contributes $1$ to $local\mbox{-}error^{\eta}(Y_{r}^{\eta},X_{r}^{\eta})$. In the second case,
$x$ can again be charged to $y$ which contributes $1$ to $local\mbox{-}error^{\eta}(Y_{s}^{\eta},X_{s}^{\eta})$.

Therefore, since $x$ cannot belong to multiple $R_{a}^j, T_{a}^{j}$, we have

$$\sum_{j=l}^{1}\sum_{a=g^{j}}^{h^{j}}|E(R_{a}^j,T_{a}^j)| \leq \sum_{j=l}^{1}\sum_{a=g^{j}}^{h^{j}}local\mbox{-}error^j(Y_a^j,X_a^j)$$

We also have

$$\sum_{j=l}^{1}\sum_{a=g^{j}}^{h^{j}}|D(R_{a}^j,T_{a}^j)| \leq \sum_{j=l}^{1}\sum_{a=g^{j}}^{h^{j}}local\mbox{-}error^j(Y_a^j,X_a^j)$$

Therefore,

$$\sum_{j=l}^{1}\sum_{a=g^{j}}^{h^{j}}\stred(L_a^j\cap Y_a^j, L_a^j \cap Y_a^j)=d_{g}^{h}+\sum_{j=l}^{1}\sum_{a=g^{j}}^{h^{j}}(|D(R_{a}^j,T_{a}^j)|+|E(R_{a}^j,T_{a}^j)|)$$

$$\leq d_{g}^{h}+2 \sum_{j=l}^{1}\sum_{a=g^{j}}^{h^{j}}local\mbox{-}error^j(Y_a^j,X_a^j) \leq d_{g}^{h}+2 l d_{g}^{h}\sqrt{20\log{d_{g}^{h}}}$$

Since, this holds for any level, and $l \leq \lceil \log{z} \rceil +1$, we get the desired bound stated in the lemma. The probabilistic bound comes from the same argument as in Lemma \ref{lemma:phaseL-edit}.
\end{proof}

Therefore, we get the improved theorem

\begin{theorem*}[\ref{theorem:main}]
There exists an algorithm that obtains an $O(\beta(n)\log{z}\sqrt{\log{d}})$-approximation factor for edit distance computation  from strings  to \dy$(s)$ for any $s \geq 2$ in $O(n\log{n}+\alpha(n))$ time with probability at least $\left(1-\frac{1}{n}-\frac{1}{d}\right)$, where there exists an algorithm for {\sc StrEdit} running in $\alpha(n)$ time that achieves an approximation factor of $\beta(n)$.
\end{theorem*}

\subsection{Getting Rid of $\sqrt{\log{d}}$-term in the Approximation Factor}

We can improve the approximation factor to $O(\beta(n)\log{z})$, if we consider $O(n^{\epsilon}\log{n})$ iterations instead of $O(\log{n})$. We can then use Corollary \ref{corol:1} instead of Lemma \ref{lemma:2}
to bound the hitting time of the random walk to $\frac{1}{\epsilon}\frac{d^2}{\log{d}}$, and hence the number of deletions performed by {\bf Random-deletion} process. Local error now improves from $A_{a}^{OPT,local^i} \setminus A_{a}^{local^i}=O( d_{g}^{h}\sqrt{\log{d_{g}^{h}}})$ to $A_{a}^{OPT,local^i} \setminus A_{a}^{local^i}=O( \frac{1}{\epsilon} d_{g}^{h})$ using the same Chernoff bound argument. Now following the same argument as before and noting that the best known algorithm for \stred~also runs in $n^{1+\epsilon}$ time returning an $O((\log{n})^{\frac{1}{\epsilon}})$ approximation we get the following theorem.

\begin{theorem*}[\ref{theorem:final}]
For any $\epsilon > 0$, there exists an algorithm that obtains an $O(\frac{1}{\epsilon}\log{z}(\log{n})^{\frac{1}{\epsilon}})$-approximation factor for edit distance computation  from strings  to \dy$(s)$ for any $s \geq 2$ in $O(n^{1+\epsilon})$ time with high probability.
\end{theorem*}

\noindent{\bf Note.}
 Due to local computations, it is possible to parallelize this algorithm.

\section{Memory Checking Languages}
\label{section:mc}
Our algorithm and analysis for \dy$(s)$ gives a general framework which can be applied to the edit distance problem to many languages. Here we illustrate this through a collection of memory checking languages such as \pq, \queue, \stack~and \deque. These languages check whether a given transcript in memory corresponds to a particular data structures such as priority queue, queue, stack and double-ended queue respectively. Formally, we ask the following question, we observe a sequence of $n$ updates and queries to (an implementation of) a data structure, and we want to report whether or not the implementation operated correctly on the
instance observed and if not what is the minimum number of changes need to be done to make it a correct implementation. A concrete example is to observe a transcript of operations on a priority queue: we see
a sequence of insertions intermixed with items claimed to be the results of extractions, and the problem
is to detect a way to minimally change the transcript to make it correct. This is the model where checker is invasive and allowed to make changes. Most of the prior literature considered such invasive checkers \cite{a:02,ckm:07,dnrv:09,nr:09}.

\noindent{\bf Stack.} Let \stack($s$) denote the language over interaction sequences of $s$ different symbols that correspond to stack operations.
Let $ins(u)$ correspond to an insertion of $u$ to a stack, and $ext(u)$ is an extraction of $u$ from the stack. Then
$\sigma \in \stack$ iff $\sigma$ corresponds to a valid transcript of operations on a stack which starts and ends empty.
That is, the state of the stack at any step $j$ can be represented by a string $S^j$ so that
$S^{0}=\phi$, $S^{j}=uS^{j-1}$ if $\sigma_j=ins(u)$, $uS^{j}=S^{j-1}$ if $\sigma(j)=ext(u)$ and $S^{n}=\phi$.

It is easy to see that \stack($s$)=\dy($s$) by assigning $u$ to $ins(u)$ and $\bar{u}$ to $ext(u)$. Hence, we can employ the algorithm for \dy$(s)$
to estimate the edit distance to \stack($s$) efficiently.\\

\noindent{\bf Priority Queue.} Let \pq($s$) denote the language over interaction sequences of $s$ different symbols that correspond to priority queue operation. That is the state of priority queue at any time $j$ can be represented by a multiset $M^{j}$ such that $M^{0}=M^{n}=\emptyset$. $M^{j}=M^{j-1}u$ if $\sigma_j=ins(u)$ and $M^{j}=M^{j-1} \setminus \{\min(M^{j-1})\}$ if $\sigma_j=ext(\nu)$. We view $ins(u)$ as $u$ and $ext(u)=\bar{u}$, but each $u$ now has a priority. Note that $\sigma$ can be represented as $Y_1X_1Y_2X_2....Y_zX_z$ where each $Y_i \in \parens^{+}$ and each $X_i \in \barparens^{+}$.

We proceed with {\em Random-deletion} but when the process starts considering symbols from $X_k$, we sort the prefix of open parenthesis by priority, so that highest priority element is at the stack top. After that using the boundaries computed by {\em Random-deletion}, one can employ the main refined algorithm from Section~\ref{section:main}. It can be verified by employing the same analysis that this results in a polylog-approximation algorithm for \pq~in $\tilde{O}(n)$ time.\\

\noindent{\bf Queue.} Let \queue($s$) denote the language over interaction sequences of $s$ different symbols that correspond to queue operations. That is the state of queue at any time $j$ can be represented by a string $Q^{j}$ such that $Q^{0}=Q^{n}=\phi$. $Q^{j}=Q^{j-1}u$ if $\sigma_j=ins(u)$ and $uQ^{j}=Q^{j-1}$ if $\sigma_j=ext(u)$. Now \queue($s$) is nothing but a \pq($s$) with priority given by time of insertion, earlier a symbol is inserted, higher is its priority. Therefore using the algorithm for \pq($s$), we can estimate the edit distance to \queue($s$) efficiently.\\

\noindent{\bf Deque.} Let \deque($s$) denote the language over interaction sequences that correspond to double-ended
queues. That is, there are now two types of insert and extract operations, one operation for the head and
one for the tail. We now create two strings $\sigma_1$ and $\sigma_2$ where $\sigma_1$ contains all the insertions and only extractions from the tail, whereas
$\sigma_2$ contains again all the insertions and only extractions from the head. $\sigma_1$ is created according to \stack~ protocol, whereas $\sigma_2$ is created according to \queue~protocol. We start running {\em Random-deletion} on both $\sigma_1$ and $\sigma_2$ simultaneously as follows. If {\em Random-deletion} is comparing (may lead to either matching or deletion) an extraction $\sigma_j$ in $\sigma_1$ and $\sigma_{j'}$ in $\sigma_2$, and $j < j'$, then we take one step of {\em Random-deletion} in $\sigma_1$ and if $j' < j$ then we take one step of {\em Random-deletion} in $\sigma_2$. If {\em Random-deletion} deletes an insertion from $\sigma_1$ which still exists in $\sigma_2$, we delete it from $\sigma_2$ as well and vice-versa. Once, we can perform {\em Random-deletion}, we can employ our main algorithm and the same analysis to show that this gives an $O(poly\log{n})$ approximation. 

\section{Conclusion}
In this paper, we give the first nontrivial approximation algorithm for edit distance computation to \dy$(s)$. \dy$(s)$ is a fundamental context free language, and the technique developed
here is general enough to be useful for a wide range of grammars. We illustrated this through considering languages accepted by common data structures. Is it possible to characterize this general class of grammars for
which the method can be applied ? What happens when there is a space restriction ? From the hardness point of view, it is known a nondeterministic version of \dy$(2)$ is the hardest context free grammar, and parsing any arbitrary context free grammar
is as hard as boolean matrix multiplication. It will be an interesting result to show that exact computation of edit distance to \dy$(s)$ also requires time same as boolean matrix multiplication.
Finally, It would be good to get rid off $\log{z}$ term in the approximation factor, and otherwise establish a gap between approximation hardness of string and \dy$(s)$ edit distance problem.

\section{Acknowledgement}
The author would like to thank Divesh Srivastava and Flip Korn for asking this nice question which led to their joint work \cite{kssy:13}, Arya Mazumdar for many helpful discussions and Mikkel Thorup for comments at an early stage of this work.
\bibliographystyle{plain}
\bibliography{cluster}
\section{Appendix}
All missing proofs are provided here.

 \begin{lemma*}[\ref{lemma:deletion}]
 For any string $\sigma \in (\parens \cup \barparens)^{*}$, $OPT(\sigma) \leq OPT_{d}(\sigma) \leq 2 OPT(\sigma)$.
 \end{lemma*}

 \begin{proof}
 Clearly, $OPT(\sigma) \leq OPT_{d}(\sigma)$, since to compute $OPT(\sigma)$, all edit operations: insertion, deletion and substitution were allowed but
 for $OPT_{d}(\sigma)$ only deletion was allowed and hence the number of edit operations can only increase.

 To prove the other side of inequality, consider each type of edits that are done on $\sigma$ to compute $OPT(\sigma)$. First consider only the insertions. Let the positions of insertions are immediately after the indices $i_1, i_2, ...i_l$. These insertions must have been done to match symbols, say at positions $j_1, j_2, ..., j_l$, otherwise, it is easy to refute that $OPT(\sigma)$ is not optimal. Instead of the $l$ insertions, we could have deleted the symbols at $j_1, j_2, ..., j_l$ with equal cost. Therefore, all insertions can be replaced by deletions at suitable positions without increasing $OPT(\sigma)$.

 Next, consider the positions where substitutions have happened. Let these be $i'_{1}, i'_{2},..., i'_{l'}$, for some $l' \geq 0$. If $l'=0$, then $OPT(\sigma)=OPT_{d}(\sigma)$. Otherwise, let $l'\geq 1$. Consider the position $i'_{1}$. After substitution at position $i'_{1}$, the new symbol at $i'_{1}$  must match a symbol at some  position $j$, such that either $j \in \{i'_2, i'_3, ...i'_{l'}\}$ or $j$ is outside of this set. When $j \in \{i'_2, i'_3, ...i'_{l'}\}$, instead of two substitutions at $i'_1$ and $j$, one can simply delete the symbols at these two positions maintaining the same edit cost. When $j$ does not belong to $\{i'_2, i'_3, ...i'_{l'}\}$, instead of one substitution at position $i'_{1}$, one can do two deletions at positions $i'_{1}$ and $j$. Therefore each substitution operation can be replaced by at most two deletions. Whereas each insertion can be replaced by a single deletion. Continuing in this fashion, overall, this ensures a $2$ approximation, that is $OPT_{d}(\sigma) \leq 2 OPT(\sigma)$.
 \end{proof}

 \begin{lemma*}[\ref{lemma:opt-stack}]
There exists an optimal algorithm that makes a single scan over the input pushing open parenthesis to stack and when it observes a close parenthesis, it either pops the stack top, if it matches the observed close parenthesis and removes both from further consideration, or edits (that is deletes) either the stack top or the observed close parenthesis whenever there is a mismatch.
\end{lemma*}
\begin{proof}
The statement is true when there is $0$ error. Suppose the statement is true when the number of minimum edits required is $d$. Now consider any string for which the minimum edit distance to \dy~ language is $d+1$. A stack based algorithm must find a mismatch at least once between stack top and the current symbol in the string. Consider the first position where it happens. Suppose, without loss of generality, at that point, the stack top contains $``(''$ and the current symbol is $``]''$. Any optimal algorithm that does minimum number of edits must change at least one of these symbols. There are only two alternatives, (a) delete $``(''$, or (b) delete $``]''$. Our stack based algorithm can take exactly the same option as the optimal algorithm. This reduces the number of edits required in the remaining string to make it well-balanced and the induction hypothesis applies.
\end{proof}
\subsection{Pseudocode of Refined Algorithm}
\begin{algorithm}[h!]
\caption{Improved Edit Distance Computation to \dy$(s)$}
\label{alg2}
 \begin{algorithmic}
 \STATE Input $\sigma=Y_1X_1Y_2X_2....Y_zX_z$
 \STATE Initialization: $\sigma'=\sigma$; $i=1$
\WHILE{ $a \leq z$}
  \STATE $N_{a,min}=\infty$; $Z_{a,min}=\emptyset$, $startIn=i$\;
  \FOR{$count=1; count < 2\log{n}; count++;$}
    \STATE $Z=\emptyset$, $d_a=0$, $i=startI$
    \WHILE{processing on $X_a$ is not completed}
        \IF{$\sigma'[i]==o$}
                \STATE Insert $\sigma'[i]$ in stack
                \STATE $i++$
        \ELSIF{$\sigma'[i]$ matches top of stack}
                \STATE match $\sigma'[i]$ with top of stack, append top of stack to $Z$, and remove both of them
                \STATE $i++$
        \ELSE
                \STATE with probability $\frac{1}{2}$ each select one of $\sigma'[i]$ or top of stack to be deleted
                    \IF{top of stack is selected}
                        \STATE append that to $Z$
                    \ENDIF
                \STATE delete the selected symbol
                \STATE $d_a=d_a+1$
                \IF{$\sigma'[i]$ is deleted}
                     \STATE $i++$
                \ENDIF
        \ENDIF
    \ENDWHILE
    \IF {$d_a \leq N_{a,min}$}
        \STATE $N_{i,min}=d_a$; $Z_{a,min}=Z$, $endIn=i-1$\;
    \ENDIF
    \STATE Start again with $\sigma'$\;
 \ENDFOR
    \STATE Remove $(Z_{a,min},X_a)$ from $\sigma'$\;
    \STATE $(R_{a},T_{a})=\stred(Z_{a,min},X_{a})$\;
    \STATE Replace $(Z_{a,min},X_{a,min})$ in $\sigma$ by $(R_{a},T_{a})$\;
    \STATE $a=a+1$, $i=endIn+1$\;
\ENDWHILE
 \IF{ there are excess open parenthesis in $\sigma'$}
    \STATE Delete those corresponding open parenthesis from $\sigma$;\
 \ENDIF
 \RETURN{ $\sigma$}
 \end{algorithmic}
 \end{algorithm}
%\section*{Appendix}
%Here we give all the omitted proofs and description of the algorithms under the same section heading.
%\input{preli-appen}
%\input{algo1-appen}
\end{document}